\documentclass[11pt,a4paper]{article}
\usepackage[utf8]{inputenc}
\usepackage[T1]{fontenc}
\usepackage{amsmath}
\usepackage{amsfonts}
\usepackage{amssymb}
\usepackage{amsthm}
\usepackage[
left=3.5cm, right=3.5cm, top=3.3cm, bottom=3.3cm
]{geometry}
\usepackage{authblk}
\usepackage[inline]{enumitem}
\usepackage{multicol}
\usepackage{stmaryrd}
\usepackage{bm}
\usepackage{mathrsfs}
\usepackage{textcomp} 
\usepackage{xcolor}
\usepackage{graphicx}
\usepackage{doi}
\usepackage{hyperref}
\usepackage{tabularx}

\usepackage{tikz}
	\usetikzlibrary{positioning}

\newtheorem{theorem}{Theorem}
\newtheorem{lemma}{Lemma}
\newtheorem{proposition}{Proposition}
\newtheorem{observation}{Observation}
\newtheorem{definition}{Definition}

	\theoremstyle{definition}

	\theoremstyle{remark}

\makeatletter  
\def\th@plain{%
  \thm@notefont{}
  \itshape 
}
\def\th@definition{%
  \thm@notefont{}
  \normalfont 
}
\makeatother

\newcommand{\f}{\phi}
\newcommand{\ff}{\psi}

\newcommand{\F}{\mathsf{F}}

\newcommand{\leftd}{\scalebox{0.9}{\normalfont\texttt{\textlangle}}}
\newcommand{\rightd}{\scalebox{0.9}{\normalfont\texttt{\textrangle}}}
\newcommand{\leftb}{\scalebox{0.9}{\normalfont\texttt{[}}}
\newcommand{\rightb}{\scalebox{0.9}{\normalfont\texttt{]}}}

\newcommand{\B}[1]{\leftb #1 \rightb}
\newcommand{\D}[1]{\leftd #1 \rightd}

\newcommand{\Skip}{{\normalfont\textbf{skip}}}
\newcommand{\Abort}{{\normalfont\textbf{abort}}}
\newcommand{\If}{{\normalfont\textbf{if }}}
\newcommand{\Then}{{\normalfont\textbf{ then }}}
\newcommand{\Else}{{\normalfont\textbf{ else }}}
\newcommand{\While}{{\normalfont\textbf{while }}}
\newcommand{\Do}{{\normalfont\textbf{ do }}}
\newcommand{\Cond}[3]{\If #1 \Then #2 \Else #3}
\newcommand{\Loop}[2]{\While #1 \Do #2}

\newcommand{\ant}{\bot}
\newcommand{\dom}{\top}

\newcommand{\K}{\mathsf{K}}
\newcommand{\KA}{\mathsf{KA}}
\newcommand{\KAT}{\mathsf{KAT}}
\newcommand{\dKA}{\mathsf{dKA}}
\newcommand{\aKA}{\mathsf{aKA}}
\newcommand{\dKAT}{\mathsf{dKAT}}
\newcommand{\aKAT}{\mathsf{aKAT}}
\newcommand{\PDL}{\mathsf{PDL}}
\newcommand{\RTA}{\mathsf{RTA}}

\newcommand{\EXPTIME}{\mathrm{EXPTIME}}
\newcommand{\PSPACE}{\mathrm{PSPACE}}

\newcommand{\tot}{\leftrightarrow}

\newcommand{\Eq}{\equiv}

\newcommand{\EQ}[1]{\overset{\scriptscriptstyle\ensuremath{#1}}{\equiv}}

\newcommand{\nequiv}{\not\equiv}

\newcommand{\xto}{\xrightarrow}
\makeatletter
\newcommand{\xRightarrow}[2][]{\ext@arrow 0359\Rightarrowfill@{#1}{#2}}
\makeatother
\let\xTo\xRightarrow
\let\hat\widehat
\newcommand{\tin}{\triangleleft}
\newcommand{\bisim}{\mathrel{\underline{\tot}}}

\newcommand\dunderline[3][-2pt]{{%
  \sbox0{#3}%
  \ooalign{\copy0\cr\rule[\dimexpr#1-#2\relax]{\wd0}{#2}}}}
\newcommand{\uline}[1]{\dunderline[-2pt]{1pt}{#1}}
\let\emph\uline

\title{Kleene Algebra with Dynamic Tests: Completeness and Complexity}
\author{Igor Sedl\'ar}
\affil{The Czech Academy of Sciences, Institute of Computer Science 
	\authorcr Pod Vod\'arenskou v\v{e}\v{z}\'i 271/2, Prague, The Czech Republic}
	
\begin{document}
\maketitle

\begin{abstract}
We study versions of Kleene algebra with dynamic tests, that is, extensions of Kleene algebra with domain and antidomain operators. We show that Kleene algebras with tests and Propositional dynamic logic correspond to special cases of the dynamic test framework. In particular, we establish completeness results with respect to relational models and guarded-language models, and we show that two prominent classes of Kleene algebras with dynamic tests have an $\EXPTIME$-complete equational theory. 
\end{abstract}

\section{Introduction}

Kleene algebra with tests ($\KAT$) is a well-known algebraic framework for equational reasoning about program equivalence and correctness \cite{Kozen1997,KozenSmith1997}. $\KAT$ adds to Kleene algebra ($\KA$) \cite{Kozen1994} a Boolean algebra of tests together with a negation operator that applies only to tests. Various extensions of $\KA$ and $\KAT$ with more general (double) negation operators were considered in \cite{DesharnaisEtAl2004,DesharnaisEtAl2006,DesharnaisStruth2008,DesharnaisStruth2011}, for example. These operators include the \textit{antidomain} operator, related to dynamic negation of Dynamic Predicate Logic \cite{GroenendijkStokhof1991}, and the \textit{domain} operator, equivalent to double antidomain. Intuitively, both operators form tests out of programs --- or \textit{dynamic tests} ---  where the antidomain (domain) of a program is the test if the program diverges (halts).

In contrast to $\KA$ and $\KAT$, completeness and complexity results for their extensions with dynamic tests are largely missing.\footnote{However, already \cite{EhmEtAl2003} note that the equational theory of ``separable'' $\KAT$ with domain is $\EXPTIME$-complete. Our result is stronger since it does not hinge on separability.} We fill this gap by proving relational completeness and (parametrized) guarded-language completeness results for $\dKA$ ($\KA$ with domain), $\aKA$ ($\KA$ with both domain and antidomain), $\dKAT$ ($\KAT$ with domain) and $\aKAT$ ($\KAT$ with both domain and antidomain), and $\EXPTIME$-completeness results for $\aKAT$ and $\aKA$.

A central role in our results is played by $\aKAT$, which can be seen as an   algebraic counterpart of Propositional Dynamic Logic ($\PDL$) \cite{FischerLadner1979,HarelEtAl2000}, providing a one-sorted alternative to two-sorted Dynamic algebras \cite{Kozen1979,Pratt1991a,Pratt1979} or Test algebras \cite{TrnkovaReiterman1987}. $\KA$, $\KAT$ and their extensions with dynamic tests can all be seen as syntactic fragments of $\aKAT$. Hence, our completeness results for $\aKAT$ yield similar results for other extensions for free. The close relation of $\aKAT$ and $\PDL$ is used to establish the $\EXPTIME$-completeness result for the former. We show, in addition,  that $\aKA$ has a distinguished position in the family in the sense that all the other members, including $\aKAT$, embed into $\aKA$. This shows that $\aKA$ is $\EXPTIME$-complete as well.

This paper is an extended version of the conference paper \cite{Sedlar2023}. The main result of the conference paper is the $\EXPTIME$-completeness result for $\aKA$, established by constructing a mutual embedding between $\aKA$ and the equational theory of Relational test algebras ($\RTA$) studied by Marco Hollenberg \cite{Hollenberg1998}. Hollenberg's completeness result for $\RTA$ is a crucial lemma in \cite{Sedlar2023}. Another result reported in \cite{Sedlar2023} is that the equational theory of $\aKA$ coincides with the equational theory of $^{*}$-continuous $\aKA$. The proofs in \cite{Sedlar2023} also show that the equational theory of $\aKA$ coincides with the equational theory of relational $\aKA$, although this is not noted in the paper. The present paper bypasses $\RTA$ and applies Hollenberg's relational completeness proof strategy directly to $\aKAT$, obtaining a (parametrized) guarded-language completeness result for $\aKAT$ along the way. Hollenberg's strategy is also somewhat simplified here. The scope of the present paper is also more general than those of \cite{Hollenberg1998} and \cite{Sedlar2023}. We also strengthen the following relational completeness results: (i) Bochman and Gabbay \cite{BochmanGabbay2012} formulate a sequent system (generalizing Kozen and Tiuryn's \cite{KozenTiuryn2003} substructural logic of partial correctness $\mathsf{S}$) that is sound and complete with respect to a fragment of the equational theory of relational $\aKAT$, (ii) Hollenberg \cite{Hollenberg1997} gives an axiomatization of the equational theory of $^{*}$-free relational $\aKAT$, and (iii) McLean \cite{McLean2020} (based on the work of Mbacke \cite{Mbacke2018}) establishes a relational completeness result for $^{*}$-free relational $\dKA$.

The rest of the paper is organized as follows. Section \ref{sec:KADT} introduces Kleene algebras with dynamic tests, with focus on languages and equational theories. Section \ref{sec:REL-models} introduces relational models for the language of  Kleene algebras with dynamic tests and Section \ref{sec:PDL} notes that $\aKAT$, the strongest version of Kleene algebra with dynamic tests we consider, is ``$\PDL$ in disguise''. Section \ref{sec:LAN-models} discusses models based on guarded strings and Section \ref{sec:LAN} proves the language completeness result. Relational completeness is established in Section \ref{sec:REL}, along with the corollary that the equational theory of $\aKAT$ is $\EXPTIME$-complete. Section \ref{sec:FRA} discusses completeness for syntactic fragments of the equational theory of $\aKAT$ and proves an $\EXPTIME$-completeness result for $\aKA$. The final section concludes the paper and points out some interesting problems for future research.

\section{Dynamic tests}\label{sec:KADT}

In this section we introduce Kleene algebras with dynamic tests. In particular, we discuss their language and some of its fragments (Sect.~\ref{sec:KADT-language}) and we define a congruence on the language that represents their equational theory (Sect.~\ref{sec:KADT-equational}).

\subsection{Languages}\label{sec:KADT-language}

Let $\Sigma$ be an ``action alphabet'' and $\Pi$ a ``propositional alphabet'', disjoint from $\Sigma$. For the sake of simplicity, both alphabets are assumed to be countably infinite. Both alphabets will be fixed throughout the paper.

\begin{definition}
The set $\mathbb{E}$ of \emph{regular expressions with dynamic tests} (over $\Sigma$ and $\Pi$) is defined by the following grammar:
\[
e, f := \mathtt{a} \in \Sigma \mid \mathtt{p} \in \Pi \mid \mathtt{0} \mid \mathtt{1} \mid e + f \mid e \cdot f \mid e^{\ant} \mid e^{\dom} \mid e^{*}
\]
\end{definition}

We will often refer to regular expressions with dynamic tests just as \textit{expressions}. Our definition of an expression uses the regular operators $+$ (sum), $\cdot$ (multiplication, concatenation) and $\,^{*}$ (Kleene star), together with constants $\mathtt{0}$ (zero, annihilator) and $\mathtt{1}$ (one, multiplicative identity). The dynamic test operators are $\,^{\ant}$ (antidomain) and $\,^{\dom}$ (domain). For technical reasons explained shortly, we take both as primitive.\footnote{We diverge from the notation used in the literature. In some papers on Kleene algebra with domain, e.g.~\cite{DesharnaisStruth2011}, antidomain is denoted as \textsf{a} in prefix notation and domain as \textsf{d}. Desharnais et al. \cite{DesharnaisEtAl2004} use $\ulcorner$ in prefix notation for domain and in \cite{DesharnaisEtAl2006} they use $\delta$. We consider $\,^{\ant}$ and $\,^{\dom}$ to be a better fit with the Kleene algebra tradition, and a good choice when it comes to comparisons with $\PDL$. The notation $\,^{\ant}$ for antidomain is also related to the interpretation of antidomain in terms of divergence -- recall that $\bot$ is sometimes used as a symbol denoting the ``absurd state'' to which all the diverging computations lead.} We will see in the next section that $\Sigma$ and $\Pi$ will have different semantics: while $\Sigma$ will be seen a set of atomic actions, $\Pi$ will be seen as a set of atomic propositions.

\begin{definition}
The set $\mathbb{F}$ of \emph{formulas} (over $\Sigma$ and $\Pi$) is defined by the following grammar:
\[ \f, \ff := \mathtt{p} \in \Pi \mid \mathtt{0} \mid \mathtt{1} \mid \f + \ff \mid \f \cdot \ff \mid e^{\ant} \mid e^{\top} \, ,\] where $e \in \mathbb{E}$.
\end{definition}

The notion of a \emph{subexpression} is defined as expected. A \emph{subformula} of an expression $e$ is a subexpression of $e$ which is also a formula. Let $\mathrm{Sf}(e)$ be the set of subformulas of $e \in \mathbb{E}$ and let $\mathrm{Sf}(E) = \bigcup_{e \in E} \mathrm{Sf}(e)$. For $\Gamma \subseteq \mathbb{E}$, we define $\Gamma^{\bot} = \{ e^{\bot} \mid e \in \Gamma \}$ and $\Gamma^{\pm} = \Gamma \cup \Gamma^{\bot}$.

\begin{definition}
An expression $e \in \mathbb{E}$ is \emph{testable} iff it does not contain occurrences of $\,^{\dom}$. A \emph{test} is an expression $e^{\dom}$ where $e$ is testable. A \emph{parameter} is either a test or an element of $\Pi$.
\end{definition}
\noindent
Note that every test (parameter) is a formula and if $\f$ is a test (parameter) then $\f^{\ant}$ is not a test (parameter).

\begin{definition}
Various \emph{fragments of $\mathbb{E}$} are defined in Figure \ref{fig:languages}: $\mathbb{E}_{\K}$ is the smallest set of expressions that contains the base as a subset, is closed under the regular operators, and
\begin{itemize}
\item contains $e^{\dom}$ if $e$ is in the set indicated under $\dom$
\item contains $e^{\ant}$ if $e$ is in the set indicated under $\ant$
\end{itemize}
We assume that $\Phi$ is a set of parameters extending $\Pi$. A \emph{$\K$-expression} is any $e \in \mathbb{E}_{\K}$.
\end{definition}

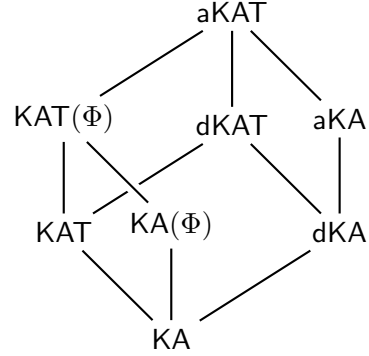
\begin{figure}
\centering
\begin{minipage}[t]{0.5\linewidth}\centering
\renewcommand{\arraystretch}{1.3}
\noindent
\begin{tabularx}{0.9\textwidth}{
	|	>{\centering\arraybackslash}X
	|	>{\centering\arraybackslash}X
		>{\centering\arraybackslash}X
		>{\centering\arraybackslash}X	
	|	
	}
	\hline
$\mathsf{K}$	& Base	 &	 $\dom$	& $\bot$ \\ \hline
$\KA$				& $\Sigma$					
						&	$\emptyset$			
						&	$\emptyset$	 	\\
$\KA(\Phi)$	& $\Sigma \cup \Phi$	
						&	$\emptyset$
						&	$\emptyset$		\\
$\dKA$			& $\Sigma$
	 					& 	 $\mathbb{E}_{\dKA}$
	 					&	$\emptyset$		\\					  
$\aKA$			& $\Sigma$
		   				& 	$\mathbb{E}_{\aKA}$
		   				&	$\mathbb{E}_{\aKA}$	\\
$\KAT$			& $\Sigma \cup \Pi$
						&	$\emptyset$
						&	 $\Pi$		\\
$\KAT(\Phi)$	& $\Sigma \cup \Phi$
						&	$\emptyset$
						&	$\Phi$	\\
$\dKAT$			& $\Sigma \cup \Pi$
						&	$\mathbb{E}_{\dKAT}$
						&	 $\Pi$		\\
$\aKAT$			& $\Sigma \cup \Pi$	
						& $\mathbb{E}_{\aKAT}$		
						& $\mathbb{E}_{\aKAT}$	\\ \hline
\end{tabularx}
\end{minipage}
\hspace*{7mm}
\noindent
\begin{minipage}[c]{0.3\linewidth}\centering
\begin{tikzpicture}[
	node distance=2cm,
	every node/.style=
		{
		inner sep=3pt}
]
\node (0,0) (KA) {$\KA$};
\node[above of=KA,yshift=-5mm] (KAp) {$\KA(\Phi)$};
\node[above left of=KA] (KAT) {$\KAT$}; 
\node[above of=KAT,yshift=-5mm] (KATp) {$\KAT(\Phi)$};
\node[above right of=KA,xshift=8mm] (dKA) {$\dKA$}; 
\node[above of=dKA,yshift=-5mm] (aKA) {$\aKA$};
\node[above right of=KAT,xshift=8mm] (dKAT) {$\dKAT$}; 
\node[above of=dKAT,yshift=-5mm] (aKAT) {$\aKAT$}; 

\path[thick] 
	(KA) edge (KAT)
	(KAT) edge (dKAT)
	(dKAT) edge (aKAT)
	(KA) edge (dKA)
	(dKA) edge (dKAT)
	(dKA) edge (aKA)
	(aKA) edge (aKAT)
	(KAT) edge (KATp)
	(KA) edge (KAp)
	(KATp) edge (aKAT)
;
\path[line width=5pt,white]
	(KAp) edge (KATp);
\path[thick]
	(KAp) edge (KATp);	

\end{tikzpicture}
\end{minipage}\caption{Fragments of the language of regular expressions with dynamic tests.}\label{fig:languages}
\end{figure}

\noindent
These fragments can be seen as being defined by a restricted version of the grammar defining $\mathbb{E}$. For instance, $\mathbb{E}_{\dKA}$ is defined by the grammar
\[ e, f := \mathtt{a} \in \Sigma \mid \mathtt{0} \mid \mathtt{1} \mid e + f \mid e \cdot f \mid e^{\top} \mid e^{*}\] whereas $\KAT(\Phi)$ is defined by (recall that $\Phi^{\pm} = \Phi \cup \{ e^{\bot} \mid e \in \Phi \}$)
\[ e, f := \mathtt{a} \in \Sigma \mid f \in \Phi^{\pm} \mid \mathtt{0} \mid \mathtt{1} \mid e + f \mid e \cdot f \mid e^{*}\]
The set $\mathbb{E}_{\KA(\Phi)}$ can be seen as the set of regular expressions over $\Sigma \cup \Phi$. Similarly, $\mathbb{E}_{\KAT(\Phi)}$ can be seen as the set of regular expressions over $\Sigma \cup \Phi^{\pm}$. We will sometimes denote the former as $\mathbb{RE}(\Phi)$ and the latter as $\mathbb{RE}(\Phi^{\pm})$. It can be shown that all fragments $\mathbb{E}_{\K}$ as well as $\mathbb{RE}(\Phi^{\pm})$ admit a well-defined complexity measure allowing proofs by structural induction.\footnote{In the case of $\mathbb{RE}(\Psi)$ for $\Psi \in \{ \Phi, \Phi^{\pm}\}$, we may define $Co: \mathbb{RE}(\Psi) \to \mathbb{N}$ such that $Co(\mathtt{a}) = Co(\f) = 0$ for $\mathtt{a} \in \Sigma$ and $\f \in \Psi$, and $Co(\sharp(e_1, \ldots, e_n)) = 1 + \sum_{i \leq n} Co(e_i)$ for $\sharp \in \{ \mathtt{0}, \mathtt{1}, +, \cdot, \,^{*} \}$. This works since $\sharp \notin \{ \dom, \ant \}$.}

$\KA$-expressions are regular expressions over $\Sigma$, and $\KAT$-expressions are a variant of the expressions used in Kleene algebra with tests \cite{Kozen1997}. In our formulation, negation (antidomain) applies only to atomic propositions $\mathtt{p} \in \Pi$. Since propositions (tests) form a Boolean algebra in $\KAT$, we can do without negations of complex Boolean formulas. $\KA(\Phi)$-expressions and $\KAT(\Phi)$-expressions will serve a technical role in our proofs later. $\dKAT$-expressions provide a fitting language for ``Kleene algebras with domain'' of \cite{EhmEtAl2003,MoellerStruth2006,DesharnaisEtAl2006}, also called ``modal Kleene algebras'' in \cite{DesharnaisEtAl2004}. $\dKA$-expressions are a language for ``Kleene algebras with domain'' as defined in \cite{DesharnaisStruth2008,DesharnaisStruth2011} (which may also be called  ``domain Kleene algebras'' based on the terminology used in these two papers), and $\aKA$-expressions are a language for ``Boolean domain Kleene algebras'' as defined in \cite{DesharnaisStruth2008,DesharnaisStruth2011} (called ``Kleene algebras with domain'' in \cite{Struth2016} and ``antidomain Kleene algebras'' in \cite{GomesStruth2016}).

Expressions of the full $\aKAT$-language inherit from $\KAT$-expressions the ability to represent basic imperative programming constructs. The constant $\mathtt{1}$ represents the program $\Skip$ (``do nothing''), $\mathtt{0}$ represents $\Abort$ (``abort computation'') and multiplication corresponds to sequential composition of programs (one may read $e \cdot f$ as ``do $e$ and, after it terminates, do $f$''). Conditionals and while loops can be represented as follows:
\[ \Cond{\f}{e}{f} \text{ as } (\f \cdot e) + (\f^{\ant} \cdot f)
\qquad
\Loop{\f}{e} \text{ as } (\f \cdot e)^{*} \cdot \f^{\bot}\]
As we shall see in the next section, it is natural to read $e^{\ant}$ as the statement ``program $e$ diverges'' and $e^{\dom}$ as ``program $e$ halts''. Hence, $\aKAT$ is able to represent imperative programs where conditions in conditionals and while loops refer to halting or divergence of other programs.


\begin{definition}\label{def:FL}
A set of parameters $\Gamma$ is \emph{Fischer--Ladner closed} (FL-closed) iff
\begin{itemize}
\item $\mathtt{p}^{\dom} \in \Gamma$ only if $\mathtt{p} \in \Gamma$
\item $\mathtt{1}^{\top}, \mathtt{0}^{\top} \in \Gamma$
\item $(\f + \ff)^{\top} \in \Gamma$ only if $\f^{\top} \in \Gamma$ and $\ff^{\top} \in \Gamma$
\item $(\f \cdot \ff)^{\top} \in \Gamma$ only if $\f^{\top} \in \Gamma$ and $\ff^{\top} \in \Gamma$
\item $(e^{*})^{\top} \in \Gamma$ only if $e^{\top} \in \Gamma$
\item if $e \neq (f \cdot \mathtt{1})$, then $e^{\top} \in \Gamma$ only if $(e \cdot \mathtt{1})^{\top} \in \Gamma$ 
\item $e^{\bot\top} \in \Gamma$ only if $(e \cdot \mathtt{1})^{\top} \in \Gamma$
\item $((e + f) \cdot \f)^{\top}$ only if $(e \cdot f)^{\top} \in \Gamma$ and $(f \cdot \f)^{\top} \in \Gamma$
\item $((e \cdot f) \cdot \f)^{\top} \in \Gamma$ only if $(f \cdot \f)^{\top} \in \Gamma$ and $(e \cdot (f \cdot \f)^{\ant\ant})^{\top} \in \Gamma$
\item $(e^{*} \cdot \f)^{\top} \in \Gamma$ only if $(e \cdot (e^{*} \cdot \f)^{\ant\ant})^{\top} \in \Gamma$
\end{itemize}
The \emph{Fischer--Ladner closure} (FL-closure) of a set of tests $\Gamma$ is the smallest Fischer--Ladner closed $\Gamma ' \supseteq \Gamma$.
\end{definition}

We note that the FL-closure of a finite set of tests is a finite set of tests.

\subsection{Equational theory}\label{sec:KADT-equational}

This section defines the equational theory of Kleene algebras with dynamic tests. Syntactic fragments of the equational theory corresponding to the language fragments from the previous subsection are also discussed.

\begin{definition}\label{def:KADT-congruence}
Let $\Eq$ be the smallest congruence $\equiv$ on $\mathbb{E}$ such that
\begin{multicols}{2}
\noindent
\begin{align}
(e \cdot f) \cdot g &\equiv e \cdot (f \cdot g)\label{KA:first}\\
e \cdot \mathtt{1} &\equiv e \equiv \mathtt{1} \cdot e\\
(e + f) + g &\equiv e + (f + g)\\
\mathtt{0} + e &\equiv \mathtt{0} \equiv e + \mathtt{0}\\
e & \equiv e + e\\
e \cdot (f + g) & \equiv (e \cdot f) + (e \cdot g)\\
(e + f) \cdot g & \equiv (e \cdot g) + (f \cdot g)\\
\mathtt{1} + (e \cdot e^{*}) & \leqq e^{*}\label{KA:fix-R}
\end{align}
\begin{align}
\mathtt{1} + (e^{*} \cdot e) & \leqq e^{*}\label{KA:fix-L}\\
f + (e \cdot g) \leqq g & \implies e^{*} \cdot f \leqq g\label{KA:least-R}\\
f + (g \cdot e) \leqq g & \implies f \cdot e^{*} \leqq g\label{KA:least-L}\\
e^{\bot} \cdot e & \equiv \mathtt{0}\label{ax:LNC}\\
(e \cdot f)^{\bot} & \equiv (e \cdot f^{\top})^{\bot}\label{ax:locality}\\
e^{\bot} + e^{\top} & \equiv \mathtt{1}\label{ax:LEM}\\
\mathtt{p}^{\top} &\equiv  \mathtt{p}\label{ax:fix}\\
e^{\dom} & \Eq e^{\ant\ant} \label{ax:domant}
\end{align}
\end{multicols}
\noindent
for all $e, f \in \mathbb{E}$ and $\mathtt{p} \in \Pi$, where $e \leqq f$ means $e + f \equiv f$. If $e \equiv f$, the we will often say that $e$ and $f$ are \emph{equivalent}.
\end{definition}

\begin{proposition}\label{prop:useful}
The following hold:
\begin{multicols}{2}
\noindent
\begin{align}
e^{\bot\top} & \equiv e^{\bot} \label{eq:Fix}\\
\mathtt{1}^{\bot} & \equiv \mathtt{0}\\
\mathtt{0}^{\bot} & \equiv \mathtt{1}\\
(e + f)^{\top} & \equiv e^{\top} + f^{\top}\\
(e^{\top} \cdot f)^{\top} & \equiv e^{\top} \cdot f^{\top}\label{eq:DiamondTest}\\
e^{\top} \cdot e & \equiv e \label{eq:Preserver}\\
(e^{*})^{\top} & \equiv \mathtt{1}\label{eq:TopStar}\\
e \cdot f^{\top} \equiv \mathtt{0} & \iff e \cdot f \equiv \mathtt{0}\label{eq:IsZero}
\end{align}
\begin{align}
\f \cdot \bar{\f} & \equiv \mathtt{0}\label{eq:ECQ-bar}\\
\f + \bar{ \f} & \equiv \mathtt{1}\label{eq:LEM-bar}\\
\bar{\bar{\f}} & \equiv \f \label{eq:DN-bar}\\
\f \cdot e  \equiv \mathtt{0} &\iff \f \leqq e^{\bot} \label{eq:Contra}\\
\f \cdot \f & \equiv \f \label{eq:Idem}\\
\f \cdot \ff & \equiv \ff \cdot \f \label{eq:Comm}\\
\f^{\top} & \equiv \f \label{eq:FTop}\\
(e \cdot \f)^{\top} \leqq \f & \implies (e^{*} \cdot \f)^{\top} \leqq \f \label{eq:Star}\\
\f \leqq (e \cdot \f^{\bot})^{\bot} &\implies \f \leqq (e^{*} \cdot \f^{\bot})^{\bot}\label{eq:Star-2}
\end{align}
\end{multicols}
\end{proposition}
\begin{proof}
\eqref{eq:Star} is established as follows. We show first that 
\begin{equation}\label{eq:Star-help}
e \cdot \f \leqq \f \cdot (e \cdot \f) \implies e^{*} \cdot \f \leqq \f \cdot (e^{*} \cdot \f)
\end{equation}
First, note that $\f \leqq \f \cdot \mathtt{1} \cdot \f \leqq \f \cdot e^{*} \cdot \f$. Second, $e \cdot \f \leqq \f \cdot (e \cdot \f)$ entails $e \cdot \f \leqq \f \cdot e$, and so we have $e \cdot \f \cdot (e^{*} \cdot \f) \leqq \f \cdot ((e \cdot e^{*}) \cdot\f)$, which entails $e \cdot \f \cdot (e^{*} \cdot \f) \leqq \f \cdot (e^{*} \cdot \f)$ by \eqref{KA:fix-L}. By combining these two observations, we obtain 
\[e \cdot \f \leqq \f \cdot (e \cdot \f) \implies \f + e \cdot (\f \cdot (e^{*} \cdot \f)) \leqq \f \cdot (e^{*} \cdot \f)\] Then, using \eqref{KA:least-R} we obtain \eqref{eq:Star-help}. We then proceed as follows. If $(e \cdot \f)^{\top} \leq \f$, then $(e \cdot \f)^{\top} \cdot (e \cdot \f) \leqq \f \cdot (e \cdot \f)$ and so $e \cdot \f \leqq \f \cdot (e \cdot \f)$ by \eqref{eq:Preserver}. Using \eqref{eq:Star-help} we obtain $e^{*} \cdot \f \leqq \f \cdot (e^{*} \cdot \f)$, which yields $e^{*} \cdot \f \leqq \f \cdot e^{*}$. Using monotonicity of $\top$ and \eqref{ax:locality}, we infer $(e^{*} \cdot \f)^{\top} \leqq (\f \cdot (e^{*})^{\top})^{\top}$ which entails our goal $(e^{*} \cdot \f)^{\top} \leqq \f$ by \eqref{eq:TopStar} and \eqref{eq:FTop}. \eqref{eq:Star-2} follows from \eqref{eq:Star}.\footnote{I am grateful to Damien Pous for a suggestion on which this proof is based.}
\end{proof}

It is clear from \eqref{ax:domant} that every expression is equivalent to a testable expression. Let $e? := f^{\top}$ where $f$ is obtained from $e$ by replacing every occurrence of $\dom$ by $\ant\ant$. Let $\mathrm{St}(e) = \{ f ? \mid f \in \mathrm{Sf}(e)\}$ and $\mathrm{St}(E) = \bigcup_{e \in E} \mathrm{St}(e)$.

In what follows, we will use the notation $e \EQ{\K} f$ to indicate that $e, f \in \mathbb{E}_{\K}$ and $e \Eq f$. The following is immediate:

\begin{observation}\label{obs:EQ}
$e \EQ{\K} f \implies e \Eq f$ for all fragments $\K$.
\end{observation}

We note that $\EQ{\aKA}$ clearly coincides with the equational theory of ``Boolean domain Kleene algebras'' as defined in \cite{DesharnaisStruth2008,DesharnaisStruth2011} (called ``Kleene algebras with domain'' in \cite{Struth2016} and ``antidomain Kleene algebras'' in \cite{GomesStruth2016}).

\section{Relational models}\label{sec:REL-models}

This section discusses the relational semantics for regular expressions with dynamic tests.

\begin{definition}\label{def:relmodel}
A \emph{relational model} (for $\Sigma$ and $\Pi$) is a structure $M = \langle X, \mathsf{rel}_M, \mathsf{sat}_M\rangle$ where $X$ is a set, $\mathsf{rel}_M : \Sigma \to 2^{X \times X}$ and $\mathsf{sat}_M : \Pi \to 2^{X}$.
\end{definition} 

The intuition behind the relational models is that elements of the action alphabet $\Sigma$ express \textit{actions} that are associated with \textit{state transitions} on a set of states $X$, and that elements of $\Pi$ express \textit{propositions} associated with \textit{sets of states}. In particular, the fact that $\langle x, y\rangle \in \mathsf{rel}_M(\mathtt{a})$ (or $x \xto{\mathtt{a}} y$ in $M$, as we shall sometimes write) means that performing action $\mathtt{a}$ in state $x$ may result in state $y$,\footnote{We do not assume that actions are deterministic, that is, $\textsf{rel}_M(\mathtt{a})$ is not necessarily a total function.} and the fact that $x \in \textsf{sat}_M(\mathtt{p})$ means that proposition $\mathtt{p}$ is \textit{satisfied} (or \textit{true}) in state $x$.

\begin{definition}\label{def:M-inter}
 For each relational model $M = \langle X, \mathsf{rel}_M, \mathsf{sat}_M\rangle$ (for $\Sigma$ and $\Pi$) we define the \emph{$M$-interpretation} of $\mathbb{E}$ as the function $\llbracket-\rrbracket_M : \mathbb{E} \to 2^{X \times X}$ such that:
\begin{gather*}
\llbracket \mathtt{a}\rrbracket_M = \mathsf{rel}_M(\mathtt{a}) \qquad
\llbracket \mathtt{p}\rrbracket_M = 1_{\mathsf{sat}_M(\mathtt{p})} \qquad
\llbracket \mathtt{0}\rrbracket_M = \emptyset \qquad
\llbracket \mathtt{1}\rrbracket_M = 1_{X}\\[2mm]
\llbracket e + f\rrbracket_M = \llbracket e\rrbracket_M \cup \llbracket f\rrbracket_M \qquad
\llbracket e \cdot f\rrbracket_M = \llbracket e\rrbracket_M \circ \llbracket f\rrbracket_M\\[2mm]
\llbracket e^{*}\rrbracket_M = \left ( \llbracket e\rrbracket_M\right )^{*} = \bigcup_{n \geq 0} \llbracket e\rrbracket^{n}_{M}\\
\llbracket e^{\bot}\rrbracket_M = \mathsf{a} \left ( \llbracket e\rrbracket_M\right ) = \bigcup_{R \subseteq 1_{X}} \left ( R \circ \llbracket e\rrbracket_M = \emptyset \right )\\
\llbracket e^{\dom}\rrbracket_M = \mathsf{d} \left ( \llbracket e\rrbracket_M\right ) = 1_{X} \setminus \mathsf{a} (\llbracket e\rrbracket_{M})
\end{gather*}
where, for $R \subseteq 2^{X \times X}$, $R^{0}= 1_X$ and $R^{n+1} = R^{n} \circ R$. Expressions $e$ and $f$ are \emph{relationally equivalent}) iff $\llbracket e\rrbracket_M = \llbracket f\rrbracket_M$ for all $M$.
\end{definition}
\noindent
In other words,
\[ \llbracket e^{\bot}\rrbracket_M = \{ \langle x, x\rangle \mid \neg \exists y \left  ( \langle x, y\rangle \in \llbracket e\rrbracket_{M} \right )\} 
\qquad
\llbracket e^{\top}\rrbracket_M = \{ \langle x, x\rangle \mid \exists y \left  ( \langle x, y\rangle \in \llbracket e\rrbracket_{M} \right )\}
\]
Hence, in the relational setting, the antidomain operator coincides with \textit{dynamic negation} studied in Dynamic Predicate Logic \cite{GroenendijkStokhof1991,Hollenberg1997}. If a \textit{program} is seen as a binary relation on $X$ and a \textit{proposition} is seen as a subset of $X$ (or, equivalently, a subset of the identity relation on $X$), then the dynamic test operators $^{\dom}$ and $^{\bot}$ turn programs into propositions. More specifically, $R^{\dom}$ is the proposition expressing the fact that $R$ \textit{halts}: $x \in R^{\top}$ iff there is $y$ such that $\langle x, y\rangle \in R$ (there is a terminating computation of $R$ starting in $x$). Similarly, $R^{\ant}$ is the proposition expressing the fact that $R$ \textit{diverges}: $x \in R^{\ant}$ iff there is no $y$ such that $\langle x,y \rangle \in R$ (there is no halting computation of $R$ starting in $x$).

The following relational soundness result is easily verified. One of the tasks of this paper is to prove the converse; see Sect.\ref{sec:REL}.

\begin{lemma}[Relational soundness]\label{lem:R-sound}
For all $e, f \in \mathbb{E}$, if $e \equiv f$, then $\llbracket e\rrbracket_M = \llbracket f\rrbracket_M$ for all relational $M$.
\end{lemma}

\section{Complexity of deciding relational equivalence}\label{sec:PDL}

In this section we point out the close relation between relational equivalence and validity in Propositional Dynamic Logic ($\PDL$) \cite{FischerLadner1979,HarelEtAl2000}. In fact, $\aKAT$ can be seen as ``$\PDL$ in disguise''. It follows from our observations that the problem of deciding relational equivalence between arbitrary expressions is $\EXPTIME$-complete.

It is easy to prove by induction on the structure of $\f$ that $\llbracket \f\rrbracket_M \subseteq 1_{X}$ for all formulas $\f$ and all $M$. We will write
\[ (M,x) \vDash \f \iff \langle x, x \rangle \in \llbracket \f\rrbracket_M\] 

\begin{observation}
For all $\f$ and all $M$:
\begin{enumerate}
\item $(M, x) \vDash \f$ iff $\langle x, x\rangle \in \llbracket \f^{\top}\rrbracket_M$
\item $(M, x) \not\vDash \mathtt{0}$ \, and \, $(M, x) \vDash \mathtt{1}$ for all $x$
\item $(M, x) \vDash \f^{\bot}$ iff $(M, x) \not\vDash \f$
\item $(M, x) \vDash \f + \ff$ iff $(M, x) \vDash \f$ or $(M, x) \vDash \ff$
\item $(M, x) \vDash \f \cdot \ff$ iff $(M, x) \vDash \f$ and $(M, x) \vDash \ff$
\item $(M, x) \vDash (e \cdot \f)^{\ant\ant}$ iff there is $y$ such that $\langle x,y\rangle \in \llbracket e\rrbracket_M$ and $(M, y) \vDash \f$
\item $(M, x) \vDash (e \cdot \f^{\ant})^{\ant}$ iff for all $y$, $\langle x,y\rangle \in \llbracket e\rrbracket_M$ implies that $(M, y) \vDash \f$
\end{enumerate}
\end{observation}

\noindent
Hence, we can write
\begin{align*}
\f ? \quad & \text{instead of} 
 	\quad \f^{\dom}\\
\neg \f \quad & \text{instead of} 
 	\quad \f^{\bot}\\
\f \land \ff \quad & \text{instead of} 
	\quad \f \cdot \ff\\
\f \lor \ff \quad & \text{instead of} 
	\quad \f + \ff\\	
\D{e}\f \quad & \text{instead of} 
	\quad (e \cdot \f)^{\ant\ant}\\
\B{e}\f \quad & \text{instead of} 
	\quad (e \cdot \f^{\bot})^{\bot}	
\end{align*} 
and the expressions on the left hand side will have their usual semantics if seen as expressions in the standard language of $\PDL$. 


\begin{theorem}\label{thm:complexity-REL}
The problem of deciding relational equivalence between arbitrary expressions is $\EXPTIME$-complete.
\end{theorem}
\begin{proof}
Lower bound: It is shown in \cite[Ch.~8.2]{HarelEtAl2000} that the ($\EXPTIME$-hard) membership problem for polynomial-space alternating Turing machines reduces to the problem of satisfiability of PDL-formulas in relational models: For each machine $A$ and input $t$, there is a PDL-formula $\F_{A,t}$ such that $A$ accepts $t$ iff $\F_{A, t}$ is satisfiable, that is, iff there is $M$ such that $\llbracket \F_{A, t}\rrbracket_M \neq\llbracket \mathtt{0}\rrbracket_M$ iff $\F_{A,t} \nequiv \mathtt{0}$ (by Theorem \ref{thm:RC}).

Upper bound: It is sufficient to show that for every $e \in \mathbb{E}$ there is an equivalent $e'$ in the fragment corresponding to programs of $\PDL$ and that $e'$ is polynomially computable from $e$. We omit the tedious but straightforward details. It is known that the problem of deciding (relational) equivalence between PDL-programs is in $\EXPTIME$. 
\end{proof}

We note that the the $\EXPTIME$-hardness proof does not automatically carry over to fragments of the full language. We would need to check the construction of $\mathsf{F}_{A, t}$ for arbitrary $A, t$ and determine if the formula is guaranteed to belong to the given fragment. We leave this to future work.

\section{Language models}\label{sec:LAN-models}

This section introduces some notions we will use in the language completeness proof of Section \ref{sec:LAN}. In particular, we formulate a generalization of the guarded string model for Kleene algebra with tests where atoms are formed over an arbitrary set of formulas, not necessarily a set of propositional letters (Sect.~\ref{sec:LAN-models-1}). Then we introduce a variant of these models where atoms are required to be consistent (Sect.~\ref{sec:LAN-models-2}).  

\subsection{Guarded languages and Kleene algebra with tests}\label{sec:LAN-models-1}

Let $\Phi$ be a finite set of parameters which is ordered in some fixed but arbitrary way as $\f_1, \ldots, \f_n$ (all $\f_i$ are assumed to be mutually distinct). An \emph{atom} over $\Phi$ is a sequence $\ff_1 \ldots \ff_n$ where $\ff_i \in \{ \f_i, \f_i^{\ant} \}$. The set of all atoms over $\Phi$ is denoted as $\mathbb{A}(\Phi)$ and atoms are denoted by capital letters $G, H$ etc. If $\f \in \Phi$ and $G \in \mathbb{A}(\Phi)$ then we will write
\[ G \triangleleft \f \quad \text{ iff} \quad \f \text{ occurs in } G\] It is easily seen that $G \tin \f^{\ant}$ iff not $G \tin \f$, for all $\f \in \Phi$ and $G \in \mathbb{A}(\Phi)$, since if $\f$ is a parameter, then $\f^{\bot}$ is not a parameter.

A \emph{guarded string} over $\Phi$ is any sequence of the form
\[ G_1 \mathtt{a}_1 G_2 \ldots \mathtt{a}_{n-1} G_n\] where each $G_i \in \mathbb{A}(\Phi)$ and $\mathtt{a}_j \in \Sigma$. The set of all guarded strings over $\Phi$ is denoted as $\mathbb{GS}(\Phi)$. \emph{Fusion product} is a partial binary operation on $\mathbb{GS}(\Phi)$ defined as follows:
\[ xG \diamond Hy = 
\begin{cases}
xGy & G = H\\
\text{undefined} & G \neq H
\end{cases}
\]
Fusion product is lifted to \emph{$\Phi$-guarded languages} $K, L \subseteq \mathbb{GS}(\Phi)$ as expected: $L \diamond K = \{ w \diamond u \mid w \in L \And u \in K \}$. Lifted fusion product is a total operation on guarded languages.

\begin{definition}
The \emph{algebra of $\Phi$-guarded languages} is 
\[ \mathbf{GL}(\Phi) = \langle 2^{\mathbb{GS}(\Phi)}, 2^{\mathbb{A}(\Phi)}, \cup,\, \diamond, \,^{*}, \,^{\bot}, \,^{\dom}, \emptyset, \mathbb{A}(\Phi)\rangle\] where $K^{*} = \bigcup_{n \geq 0} K^{n}$ (assuming that $K^{0} = \mathbb{A}(\Phi)$ and $K^{n+1} = K^{n} \diamond K$) and 
\[ L^{\bot} = \{ G \in \mathbb{A}(\Phi) \mid  \{ G \} \diamond L = \emptyset\} 
\qquad
L^{\dom} = \{ G \in \mathbb{A}(\Phi) \mid  \{ G \} \diamond L \neq \emptyset\}
\, . \]
\end{definition}

As a special case, an atom over $\emptyset$ is the empty sequence $\epsilon$, and so $\mathbb{GS}(\emptyset) = \Sigma^{*}$. Fusion product on $\mathbb{GS}(\emptyset)$ reduces to concatenation of strings over $\Sigma^{*}$. $\mathbf{GL}(\emptyset)$ is an expansion of the algebra of formal languages over $\Sigma$ with $\,^{\bot}$ that sends every non-empty language to $\emptyset$ and $\emptyset$ to $\{ \epsilon \}$ and $\,^{\dom}$ that sends every non-empty language to $\{ \epsilon \}$ and $\emptyset$ to $\emptyset$. Algebras of guarded languages as used in Kleene algebra with tests \cite{Kozen1997,KozenSmith1997} are a special case where $\Phi$ is a finite subset of $\Pi$.

\begin{definition}
If $\Phi$ is a finite set of parameters, then the \emph{standard $\Phi$-interpretation} of $\mathbb{RE}(\Phi^{\pm})$ is the unique homomorphism $[-]_{\Phi} : \mathbb{RE}(\Phi^{\pm}) \to \mathbf{GL}(\Phi)$ such that
\[ [\mathtt{a}]_{\Phi} = \{ G \mathtt{a} H \mid G, H \in \mathbb{A}(\Phi) \}
\qquad\quad [\f]_{\Phi} = \{ G \mid G \in \mathbb{A}(\Phi) \And G \tin \f \}
\] for $\mathtt{a} \in \Sigma$ and $\f \in \Phi^{\pm}$.
\end{definition}

In the next section we will use a variant of Kozen and Smith's language completeness result for Kleene algebras with tests \cite{KozenSmith1997}: 

\begin{theorem}[Essentially Kozen and Smith \cite{KozenSmith1997}]\label{thm:KAT-completeness}
Let $\Phi$ be a finite set of parameters. For all $e, f \in \mathbb{E}_{\KAT(\Phi)}$,
\[ e \EQ{\KAT(\Phi)} f \iff [e]_{\Phi} = [f]_{\Phi}\]
\end{theorem}

Note that Theorem \ref{thm:KAT-completeness} and Observation \ref{obs:EQ} imply a form of language completeness, namely, \[ [e]_{\Gamma} = [f]_{\Gamma} \implies e \equiv f\] for all $e, f \in \mathbb{RE}(\Gamma^{\pm}) \subseteq \mathbb{E}$. The problem with this result is that it is not accompanied by a suitable \textit{soundness} result. Since both $\mathtt{a}$ and $\mathtt{a}^{\bot}$ are testable, both $\mathtt{a}^{\top}$ and $\mathtt{a}^{\bot\top}$ are tests and $\mathtt{a^{\top}} \mathtt{a}^{\bot\top} \in \mathbb{A}\left( \Gamma \right)$ for $\Gamma = \{ \mathtt{a^{\top}}, \mathtt{a}^{\bot\top} \}$. Hence, $[\mathtt{a^{\top}} \cdot \mathtt{a}^{\bot\top}]_{\Gamma} \neq \emptyset$ even if $\mathtt{a^{\top}} \cdot \mathtt{a}^{\bot\top} \equiv \mathtt{0}$. We have to pay attention to consistency of atoms. Moreover the standard interpretation of the action alphabet is an issue too: $\mathtt{a}^{\bot} \cdot \mathtt{a} \equiv \mathtt{0}$, but $[\mathtt{a}^{\bot} \cdot \mathtt{a}]_{\Gamma} = \{ G \mathtt{a} H \mid G \tin \mathtt{a}^{\bot} \} \neq \emptyset$.   

\subsection{Algebras of consistently guarded languages}\label{sec:LAN-models-2}

We will not distinguish between a non-empty sequence of expressions $e_1 \ldots e_n$ and the expression $e_1 \cdot \ldots \cdot e_n$ (assuming some fixed bracketing). An atom $G$ is \emph{consistent} iff $G \nequiv \mathtt{0}$. The set of all consistent atoms over $\Phi$ will be denoted as $\mathbb{C}(\Phi)$. A \emph{consistently guarded string} over $\Phi$ is any guarded string $G_1 \mathtt{a_1} \ldots \mathtt{a}_{n-1} G_n$ where all $G_i \in \mathbb{C}(\Phi)$. The set of all consistently guarded strings over $\Phi$ is denoted as $\mathbb{CS}(\Phi)$.\footnote{Hollenberg \cite{Hollenberg1998} uses the same notation for \textit{consistent} guarded strings. Note that a guarded string can be consistently guarded without being consistent.} A crucial role in the proof of language completeness in the next section will be played by the following variant of $\mathbf{GL}(\Phi)$, based on $\mathbb{CS}(\Phi)$:

\begin{definition}
The \emph{algebra of consistently $\Phi$-guarded languages} is 
\[ \mathbf{CL}(\Phi) = \langle 2^{\mathbb{CS}(\Phi)}, 2^{\mathbb{C}(\Phi)}, \cup,\, \diamond, \,^{*}, \,^{\bot}, \,^{\dom}, \emptyset, \mathbb{C}(\Phi)\rangle\] where $K^{*} = \bigcup_{n \geq 0} K^{n}$ (assuming that $K^{0} = \mathbb{C}(\Phi)$ and $K^{n+1} = K^{n} \diamond K$) and 
\[ L^{\bot} = \{ G \in \mathbb{C}(\Phi) \mid  \{ G \} \diamond L = \emptyset\} 
\qquad
L^{\dom} = \{ G \in \mathbb{C}(\Phi) \mid  \{ G \} \diamond L \neq \emptyset\}
\, . \]
\end{definition}

Note that $\mathbf{CL}(\Phi)$ is in general not a subalgebra of $\mathbf{GL}(\Phi)$ since $\mathbb{A}(\Phi)$, the multiplicative identity of $\mathbf{GL}(\Phi)$, is not necessarily an element of $\mathbf{CL}(\Phi)$. However, clearly $\mathbb{CS}(\Phi) \subseteq \mathbb{GS}(\Phi)$.

\begin{definition}\label{def:canonical-int}
Let $\Gamma$ be a set of tests. The \emph{canonical $\Gamma$-interpretation} of $\mathbb{E}$ is the unique homomorphism $\llbracket-\rrbracket : \mathbb{E} \to \mathbf{CL}(\Gamma)$ such that
\[ \llbracket \mathtt{a}\rrbracket_{\Gamma} = \{ G \mathtt{a} H \in \mathbb{CS}(\Gamma) \mid G \mathtt{a} H \nequiv \mathtt{0} \}
\qquad\quad \llbracket \mathtt{p}\rrbracket_{\Gamma} = \{ G \mid G \in \mathbb{C}(\Gamma) \And G \leqq \mathtt{p} \}\] for all $\mathtt{a} \in \Sigma$ and $\mathtt{p} \in \Pi$.
\end{definition}

In the cases where $\Phi \in \{ \emptyset, \Pi \}$, the above definition reduces to the usual notions of the canonical language interpretation from Kleene algebra and Kleene algebra with tests, respectively.

\section{Language completeness}\label{sec:LAN}

In this section we establish a ``parametrized'' language completeness result for Kleene algebra with dynamic tests. Our proof draws heavily on Hollenberg's \cite{Hollenberg1998} completeness proof for an axiomatization of program equations in tests algebras, which is in turn inspired by Kozen and Smith's \cite{KozenSmith1997} proof of language completeness of Kleene algebra with tests. Our proof introduces some simplifications into Hollenberg's general strategy -- for instance, we do not use consistent guarded strings, only \textit{consistently guarded} strings. This has the advantage of not requiring a proof that the salient set of guarded strings is closed under fusion product.  

\begin{lemma}[Language soundness]\label{lem:L-sound}
Let $\Gamma$ be any finite set of parameters. For all $e, f \in \mathbb{E}$:
\[ e \equiv f \implies \llbracket e\rrbracket_{\Gamma} = \llbracket f\rrbracket_{\Gamma} \, .\]
\end{lemma}
\begin{proof}
It is sufficient to show that axioms (\ref{ax:LNC}--\ref{ax:fix}) are valid under every $\Phi$-canonical interpretation. This is left to the reader as an easy exercise.
\end{proof}

For the rest of this section, fix an arbitrary finite set of parameters. We define $\mathbb{C}(\f) := \sum \{ G \in \mathbb{C}(\Gamma) \mid G \leqq \f \}$, for any $\f$.

\begin{lemma}\label{lem:LC-truth-1}
If $\f \in \Gamma$, then $\f \equiv \mathbb{C}(\f)$.
\end{lemma}
\begin{proof}
It is clear that $\mathbb{C}(\f) \leqq \f$. If $\f \not\leqq \mathbb{C}(\f)$, then $\f \cdot \mathbb{C}(\f)^{\bot} \nequiv \mathtt{0}$. But then there is a maximal consistent set of formulas $F$ such that $\f \in F$ and $\mathbb{C}(\f) \notin F$.\footnote{$F$ is a \emph{consistent} set of formulas iff there is no finite $\{ \chi_1, \ldots, \chi_n \} \subseteq F$ such that $\chi_1 \ldots \chi_n \equiv \mathtt{0}$. $F$ is a \emph{maximal consistent} set of formulas if $F$ is consistent but $\ff \notin F$ implies that $F \cup \{ \ff \}$ is not consistent. Since formulas satisfy the axioms of Boolean algebra by Proposition \ref{prop:useful}, every consistent set of formulas  is extended by a maximal consistent set by Lindenbaum's Lemma.} But then $G \notin F$ for all $G \in \mathbb{C}(\Gamma)$ such that $G \leqq \f$. Now take the sequence of tests $H = \ff_1 \ldots \ff_n$ where $\ff_i = \gamma_i$ if $\gamma_i \in F$ and $\ff_i = \gamma_i^{\bot}$ if $\gamma_i \notin F$, for all $\gamma_i \in \Gamma$. It is clear that $H$ is a consistent atom over $\Gamma$ and $H \leqq \f$ since $\f \in \Gamma$. This contradicts the properties of $F$, and so we conclude that $\f \leqq \mathbb{C}(\f)$.
\end{proof}

\begin{lemma}\label{lem:LC-truth}
For all $\f \in \Gamma$, $\llbracket \f\rrbracket_{\Gamma} = \{ G \in \mathbb{C}(\Gamma) \mid G \leqq \f \}$. 
\end{lemma}
\begin{proof}
The proof of this lemma is based on Kozen and Parikh's completeness proof for \textsf{PDL} \cite{KozenParikh1981}. See the appendix, Section \ref{a:LC-truth}. 
\end{proof}

\begin{definition}\label{def:hat}
We define $\:\hat{ }\:$ as the smallest function $\mathbb{RE}(\Gamma) \to \mathbb{RE}(\Gamma^{\pm})$ such that (for $\f \in \Gamma$ and $\mathtt{a} \in \Sigma$)
\[
\hat{\f} = \mathbb{C}(\f) \qquad\quad
\hat{\mathtt{a}} = \sum \llbracket \mathtt{a}\rrbracket_{\Gamma} \qquad\quad
\hat{\mathtt{1}} = \sum \mathbb{C}(\Gamma)
\]
and that commutes with $\mathtt{0}$, $\cdot, +$ and $\,^{*}$.
\end{definition}

\begin{lemma}\label{lem:hat-1}
For all $e \in \mathbb{RE}(\Gamma)$, $e \equiv \hat{e}$.
\end{lemma}
\begin{proof}
Induction on the complexity of $e$. The base case $e = \f \in \Gamma$ holds by Lemma \ref{lem:LC-truth-1}. The base case $e = \mathtt{a} \in \Sigma$ is established as follows:
\begin{align*}
\mathtt{a} & \equiv \mathtt{1} \cdot \mathtt{a} \cdot \mathtt{1}
	\equiv \mathbb{C}(\mathtt{1}) \cdot \mathtt{a} \cdot \mathbb{C}(\mathtt{1})\\
	& \equiv \sum \mathbb{C}(\Gamma) \cdot \mathtt{a} \cdot \sum \mathbb{C}(\Gamma)\\
	& \equiv \sum \{ G \mathtt{a} H \mid G, H \in \mathbb{C}(\Gamma) \}\\
	& \equiv \sum \{ G \mathtt{a} H \mid G, H \in \mathbb{C}(\Gamma) \And G \mathtt{a} H \nequiv \mathtt{0} \}\\
	& \equiv \sum \llbracket \mathtt{a}\rrbracket_{\Gamma} \equiv \hat{\mathtt{a}}
\end{align*}
(The second equivalence relies on Lemma \ref{lem:LC-truth-1}.) The induction step for $\mathtt{1}$ follows from Lemma \ref{lem:LC-truth-1}. The rest of the proof is easy since $\:\hat{\,}\:$ commutes with $\mathtt{0}$, $\cdot$, $+$ and $\,^{*}$. 
\end{proof}

\begin{lemma}\label{lem:GS}
$[w]_{\Gamma} = \{ w \}$ for all $w \in \mathbb{GS}(\Gamma)$.
\end{lemma}

\begin{lemma}\label{lem:hat-2}
$\llbracket e\rrbracket_{\Gamma} = [\hat{e}\,]_{\Gamma}$ for all $e \in \mathbb{RE}(\Gamma)$.
\end{lemma}
\begin{proof}
Induction on the structure of $e$, relying heavily on Lemma \ref{lem:GS}. The base case $e = \f \in \Gamma$ is established as follows (we omit the subscript $\Gamma)$:
\begin{align*}
[\hat{\f}\,] & = \left [ \sum \{ G \mid G \in \mathbb{C}(\Gamma) \And G \leqq \f \} \right ]\\
	& = \bigcup \{ [G] \mid G \in \mathbb{C}(\Gamma) \And G \leqq \f \}\\
	& = \{ G \mid G \in \mathbb{C}(\Gamma) \And G \leqq \f \} = 
	\llbracket \f \rrbracket
\end{align*}
The third equation holds thanks to Lemma \ref{lem:GS} and the last equation holds thanks to Lemma \ref{lem:LC-truth}. The base case for $e = \mathtt{a} \in \Sigma$ is established as follows:
\begin{align*}
[\hat{\mathtt{a}}\,] & = \left [ \sum\{ G \mathtt{a} H \mid G \mathtt{a} H \nequiv \mathtt{0} \} \right ]\\
	& = \bigcup \{ [G \mathtt{a} H] \mid G \mathtt{a} H \nequiv \mathtt{0} \}\\
	& = \{ G \mathtt{a} H \mid G \mathtt{a} H \nequiv \mathtt{0} \}
	= \llbracket \mathtt{a}\rrbracket
\end{align*}
The third equation holds thanks to Lemma \ref{lem:GS}. The induction step for $\mathtt{1}$ is established as follows:
\begin{align*}
[\hat{\mathtt{1}}\,] & = \left [ \sum \mathbb{C}(\Gamma)\right ]
	 = \bigcup \{ [G] \mid G \in \mathbb{C}(\Gamma \}\\
	& = \mathbb{C}(\Gamma) = \llbracket \mathtt{1}\rrbracket
\end{align*}
The third equation holds thanks to Lemma \ref{lem:GS} and the last equation holds thanks to Lemma \ref{lem:LC-truth}. The rest is established easily using the induction hypothesis since $\:\hat{\:}\,$ is assumed to commute with $\mathtt{0}$, $\cdot$, $+$ and $\,^{*}$.
\end{proof}

\begin{theorem}[Language completeness]\label{thm:LC}
Let $E \subseteq \mathbb{E}$ be finite and let $\Gamma$ be the FL-closure of $\mathrm{St}(E)$. Then, for all $e,f \in E$: 
\[e \equiv f \iff \llbracket e\rrbracket_{\Gamma} = \llbracket f\rrbracket_{\Gamma}\]
\end{theorem}
\begin{proof}
The implication from left to right follows from Lemma \ref{lem:L-sound}. The converse implication is established as follows:
\begin{center}
$\llbracket e\rrbracket_{\Gamma} = \llbracket f\rrbracket_{\Gamma}
	\:\xTo{\text{Lemma \ref{lem:hat-2}}\,}\: 
[\hat{e}\,]_{\Gamma} = [\hat{f}\,]_{\Gamma}
	\:\xTo{\text{Theorem \ref{thm:KAT-completeness}}\,}\: 
\hat{e} \EQ{\KAT(\Gamma)} \hat{f}$\\[2mm]
$\hat{e} \EQ{\KAT(\Gamma)} \hat{f} 
	\:\xTo{\text{Observation \ref{obs:EQ}}\,} \: 
\hat{e} \equiv \hat{f} 
	\:\xTo{\text{Lemma \ref{lem:hat-1}}\,}\: 
e \equiv f$
\end{center}
\end{proof}

Similarly to well known language completeness results of Kleene algebra and Kleene algebra with tests, the completeness result presented in this section shows that two expressions are equivalent iff they denote the same language in a given language model. However, the result is ``parametrized'' in the sense that it relies on expressions not denoting languages \textit{simpliciter} but via a parameter $\Gamma$. The completeness result then holds only if a right kind of parameter is chosen. 

The relation between expressions and languages used here can be seen as a generalization of the unmediated relation between regular expressions and regular languages. In our setting, expressions can be seen as denoting \textit{generalized languages}, that is, functions $\lambda : \mathbb{E} \to (2^{\mathbb{E}} \to 2^{\mathbb{GS}})$. Standard languages are a special case where $\lambda(e)$ is a constant function.

\section{Relational completeness}\label{sec:REL}

In this section we prove a relational completeness result for Kleene algebra with dynamic tests. We rely on Theorem \ref{thm:LC}. In the rest of the section, let $\Gamma$ be a fixed finite set of parameters.

\begin{definition}\label{def:cay}
We define the function $\mathsf{cay} : 2^{\mathbb{CS}(\Gamma)} \to 2^{\mathbb{CS}(\Gamma) \times \mathbb{CS}(\Gamma)}$ as follows:
\[ \mathsf{cay} (L) = \{ \langle w, w \diamond u \rangle \mid w \in \mathbb{CS}(\Gamma) \And u \in L\}\] 
\end{definition}

We define a relational model $CS(\Gamma) = \langle \mathbb{CS}(\Gamma), \mathsf{rel}_{CS(\Gamma)}, \mathsf{sat}_{CS(\Gamma)}\rangle$ where
\[ \mathsf{rel}_{CS(\Gamma)}(\mathtt{a}) = \mathsf{cay}(\llbracket \mathtt{a}\rrbracket_{\Gamma})
\qquad \mathsf{sat}_{CS(\Gamma)}(\mathtt{p}) = \{ w \mid \mathsf{last}(w) \leqq \mathtt{p} \} \]
for $\mathtt{a} \in \Sigma$ and $\mathtt{p} \in \Pi$. (Of course, $\mathsf{last}(w)$ is the last atom occurring in $w$, and we denote as $\mathsf{first}(w)$ the first one.)

\begin{lemma}\label{lem:RC-2}
For all $\f \in \Gamma$, $\llbracket \f\rrbracket_{CS(\Gamma)} = \{ \langle w, w\rangle \mid \mathsf{last}(w) \leqq \f \}$.
\end{lemma}
\begin{proof}
See the appendix, Sect.~\ref{a:RC}.
\end{proof}

\begin{lemma}\label{lem:RC-3}
For all $e \in \mathbb{RE}(\Gamma)$, 
\[ \mathsf{cay}\left( \llbracket e\rrbracket_{\Gamma} \right) = \llbracket e\rrbracket_{CS(\Gamma)} \, .\]
\end{lemma}
\begin{proof}
Induction on the complexity of $e$. We omit the subscript $CS(\Gamma)$ in the rest of the proof. The base case for $e = \f \in \Gamma$ is established using Lemma \ref{lem:LC-truth} and Lemma \ref{lem:RC-2}: $\mathsf{cay}\left( \llbracket \f\rrbracket_{\Gamma} \right) = \{ \langle w, w\rangle \mid \mathsf{last}(w) \leqq \f \} = \llbracket \f\rrbracket$. The base case for $e = \mathtt{a} \in \Sigma$ holds by definition.

The induction step for $\mathtt{1}$ follows from Lemma \ref{lem:LC-truth}: $\llbracket \mathtt{1}\rrbracket_{\Gamma} = \mathbb{C}(\Gamma)$ and so $\mathsf{cay} \left( \llbracket \mathtt{1}\rrbracket_{\Gamma} \right) = 1_{\mathbb{CS}(\Gamma)} = \llbracket \mathtt{1}\rrbracket$. The induction step for $\mathtt{0}$ is trivial.

The induction step for $e + f$ is established as follows: $\llbracket e + f\rrbracket = \llbracket e\rrbracket \cup \llbracket f\rrbracket = \mathsf{cay} (\llbracket e\rrbracket_{\Gamma}) \cup \mathsf{cay}(\llbracket f\rrbracket_{\Gamma}) = \mathsf{cay} (\llbracket e\rrbracket_{\Gamma} \cup \llbracket f\rrbracket_{\Gamma}) = \mathsf{cay}\left( \llbracket e + f\rrbracket_{\Gamma} \right)$. The induction step for $e \cdot f$ is established as follows:
\begin{align*}
\llbracket e \cdot f\rrbracket
	& = \llbracket e\rrbracket \circ \llbracket f\rrbracket\\
	& = \{ \langle w, w \diamond u\rangle \mid w \in \mathbb{CS}(\Gamma) \And u \in \llbracket e\rrbracket_{\Gamma}\}\\ & 
		\qquad \circ \{ \langle w', w' \diamond u'\rangle \mid w' \in \mathbb{CS}(\Gamma) \And u' \in \llbracket f\rrbracket_{\Gamma} \}\\
	& = \{ \langle w, w \diamond u\rangle \mid w \in \mathbb{CS}(\Gamma) \And u \in \llbracket e\rrbracket_{\Gamma}\}\\ & 
		\qquad \circ \{ \langle w \diamond u, w \diamond u \diamond u'\rangle \mid w \in \mathbb{CS}(\Gamma) \And u \in \llbracket e\rrbracket_{\Gamma} \And u' \in \llbracket f\rrbracket_{\Gamma}\}  \\
	& = 	\{ \langle w, w \diamond u \diamond u' \rangle \mid w \in \mathbb{CS}(\Gamma) \And u \in \llbracket e\rrbracket_{\Gamma} \And u' \in \llbracket f\rrbracket_{\Gamma}\}\\
	& = 	\{ \langle w, w \diamond v \rangle \mid w \in \mathbb{CS}(\Gamma) \And v \in \llbracket e \cdot f\rrbracket_{\Gamma}\}\\
	& = \mathsf{cay} (\llbracket e \cdot f\rrbracket_{\Gamma}) 	
\end{align*}

The induction step for $e^{*}$ is established as follows:
\begin{align*}
\llbracket e^{*}\rrbracket 
	& = \bigcup_{n \geq 0} \left \llbracket e \right \rrbracket^{n}
	 	=  \bigcup_{n \geq 0} \mathsf{cay} \left( \llbracket e\rrbracket_{\Gamma} \right)^{n}\\
	& = \mathsf{cay} \left( \bigcup_{n \geq 0} \llbracket e\rrbracket^{n}_{\Gamma} \right)
		= \mathsf{cay} \left( \llbracket e^{*}\rrbracket_{\Gamma} \right)
\end{align*}
\end{proof}

\begin{theorem}[Relational completeness]\label{thm:RC}
For all $e, f \in \mathbb{E}$: 
\[ e \equiv f \iff (\forall M)(\llbracket e\rrbracket_M = \llbracket f\rrbracket_M)\]
\end{theorem} 
\begin{proof}
The implication from left to right is Lemma \ref{lem:R-sound}. The converse implication is established as follows. Fix $e, f$ such that $e \nequiv f$ and let $\Gamma$ be the FL-closure of $\mathrm{Sf}(e, f)?$. We obtain $\llbracket e\rrbracket_{\Gamma} \neq \llbracket f\rrbracket_{\Gamma}$ using Theorem \ref{thm:LC} and so $\mathsf{cay} \left( \llbracket e\rrbracket_{\Gamma} \right) \neq \mathsf{cay} \left( \llbracket f\rrbracket_{\Gamma} \right)$ since $\mathsf{cay}$ is injective.\footnote{Suppose without loss of generality that $w \in L \setminus K$. Then $\langle \mathsf{first}(w), w\rangle \in \mathsf{cay}(L) \setminus \mathsf{cay}(K)$.} Then $\llbracket e\rrbracket_{CS(\Gamma)} \neq \llbracket f\rrbracket_{CS(\Gamma)}$ and we are done.
\end{proof}

\begin{theorem}\label{thm:complexity-EQ}
The problem of deciding equivalence of arbitrary expressions is $\EXPTIME$-complete.
\end{theorem}
\begin{proof}
By Theorem \ref{thm:RC}, equivalence coincides with relational equivalence. The rest follows from Theorem \ref{thm:complexity-REL}.
\end{proof}

\section{Fragments}\label{sec:FRA}

This section discusses the bearing of the results discussed in previous sections on various fragments of $\aKAT$. 

First, it is straightforward that we obtain both kinds of completeness for fragments of $\Eq$ as a corollary of Theorems \ref{thm:LC} and \ref{thm:RC}:

\begin{theorem}\label{thm:complet}
Take a finite $E \subseteq \mathbb{E}$ and $e, f \in E$. The following are equivalent:
\begin{enumerate}
\item $e \EQ{\K} f$
\item $\llbracket e\rrbracket_{\Gamma} = \llbracket f\rrbracket_{\Gamma}$ where $\Gamma$ is the FL-closure of $\mathrm{St}(E)$
\item $\llbracket e\rrbracket_M = \llbracket f\rrbracket_M$ for all relational models $M$
\end{enumerate}
\end{theorem}

Our complexity result for $\aKAT$ can be used to establish $\EXPTIME$-completeness of $\aKA$ as follows.

The upper bound of the complexity of deciding equivalence between $\aKA$-expressions follows directly from the upper bound for arbitrary expressions. The lower bound result will use the following lemma. 

\begin{lemma}\label{lem:noPi}
Let $e'$ be the result of replacing every occurrence of $\mathtt{a}_{n}$ in $e$ by an occurrence of $\mathtt{a}_{2n}$ and replacing every occurrence of $\mathtt{p}_n$ by an occurrence of $(\mathtt{a}_{2n+1})^{\top}$. Then \[ e \equiv f \iff e' \equiv f' \, .\]
\end{lemma}
\begin{proof}
Left to right: Equivalence is preserved under substitution. Moreover, clearly $\mathtt{p}' \equiv (\mathtt{p}')^{\top}$. Right to left: If $e \nequiv f$, then there is a relational model $M$ where $\llbracket e\rrbracket_M \neq \llbracket f\rrbracket_M$ (Theorem \ref{thm:RC}). We define $M'$ by taking the universe $X$ of $M$ and stipulating that
\[ \mathsf{rel}_{M'}(\mathtt{a}_m) = 
\begin{cases}
\mathsf{rel}_M (\mathtt{a}_n) & m = 2n\\
1_{\mathsf{sat}_M(\mathtt{p}_n)} & m = 2n+1
\end{cases}
\qquad
\mathsf{sat}_{M'}(\mathtt{p}) = \emptyset
\]
It can be shown by induction on $g$ that $\llbracket g\rrbracket_M = \llbracket g'\rrbracket_{M'}$. Only the base case is interesting. The base case for $\mathtt{a}_{n}$: $\llbracket (\mathtt{a}_{n})'\rrbracket_{M'} = \llbracket \mathtt{a}_{2n}\rrbracket_{M'} = \llbracket \mathtt{a}_{n}\rrbracket_M$. The base case for $\mathtt{p}_n$: $\llbracket (\mathtt{p}_{n})'\rrbracket_{M'} = \llbracket (\mathtt{a}_{2m+1})^{\top}\rrbracket_{M'} = \{ \langle x, x\rangle \mid \exists y. \langle x,y \rangle \in \llbracket \mathtt{a}_{2n+1}\rrbracket_{M'} \} = \llbracket \mathtt{p}_n\rrbracket_M$. Now clearly $\llbracket e'\rrbracket_{M'} \neq \llbracket f'\rrbracket_{M'}$ and so $e' \nequiv f'$ by Lemma \ref{lem:R-sound}.
\end{proof}

Before using this lemma in the proof of the complexity result for $\aKA$, we should pause to appreciate what it says about $\aKA$. Since every $\EQ{\K}$ trivially embeds into $\Eq$ and $\Eq$ embeds into $\EQ{\aKA}$ by Lemma \ref{lem:noPi}, we obtain the result that every $\EQ{\K}$ embeds into $\EQ{\aKA}$. Thus, $\aKA$ occupies a central position in the family of Kleene algebras with dynamic tests.

\begin{theorem}\label{thm:complexity-aKA}
The problem of deciding equivalence between $\aKA$-expressions is $\EXPTIME$-complete.
\end{theorem}
\begin{proof}
The problem is in $\EXPTIME$ since so is deciding equivalence between arbitrary expressions. The problem is $\EXPTIME$-complete thanks to Lemma \ref{lem:noPi}: deciding equivalence between arbitrary expressions can be polynomially reduced to deciding equivalence between $\aKA$-expressions. The former is $\EXPTIME$-complete by Theorem \ref{thm:complexity-REL}.
\end{proof}

\section{Conclusion}

In this paper we discussed several versions of Kleene algebra with dynamic tests, extending Kleene algebra with dynamic tests operators $^{\ant}$ (antidomain) and $\dom$ (domain). Kleene algebra with tests is a special case where antidomain applies only to Boolean variables. On the other side of the spectrum, the strongest version of Kleene algebra with dynamic tests considered here, $\aKAT$, was shown to correspond to $\PDL$. Our main results are (i) completeness results for $\aKAT$ and its syntactic fragments with respect to relational models and specific models based on guarded strings, and (ii) an $\EXPTIME$-completeness result for $\aKAT$ following from relational completeness and the relation between $\aKAT$ and $\PDL$) and for $\aKA$ (following from the fact that $\aKAT$ embeds into $\aKA$, showing that $\aKA$ occupies a central position in the family).

A number of interesting problems remain to be solved. First, are $\dKA$ and $\dKAT$ $\EXPTIME$-hard? Second, do fragments of quasi-equational theories of $\aKAT$ and its fragments where all assumptions are of the form $e \equiv \mathtt{0}$ reduce to their equational theories, similarly as in $\KA$ and $\KAT$? Third, are there natural fragments of $\aKAT$ that are stronger than $\KAT$ but still have a $\PSPACE$-complete equational theory? Fourth, echoing the links between $\KA$ ($\KAT$) and finite automata (on guarded strings), what is the natural automata-theoretic formulation of the various versions of Kleene algebra with dynamic tests considered here? 




\begin{thebibliography}{10}

\bibitem{BochmanGabbay2012}
A.~Bochman and D.~M. Gabbay.
\newblock Sequential dynamic logic.
\newblock {\em Journal of Logic, Language and Information}, 21(3):279--298,
  2012.
\newblock \href {https://doi.org/10.1007/s10849-011-9152-y}
  {\path{doi:10.1007/s10849-011-9152-y}}.

\bibitem{DesharnaisEtAl2004}
J.~Desharnais, B.~M\"o{}ller, and G.~Struth.
\newblock {Modal Kleene Algebra and Applications -- A Survey}.
\newblock {\em Journal on Relational Methods in Computer Science}, 1:93--131,
  2004.
\newblock URL:
  \url{http://www.cosc.brocku.ca/Faculty/Winter/JoRMiCS/Vol1/PDF/v1n5.pdf}.

\bibitem{DesharnaisEtAl2006}
J.~Desharnais, B.~M\"{o}ller, and G.~Struth.
\newblock Kleene algebra with domain.
\newblock {\em ACM Trans. Comput. Logic}, 7(4):798–833, oct 2006.
\newblock \href {https://doi.org/10.1145/1183278.1183285}
  {\path{doi:10.1145/1183278.1183285}}.

\bibitem{DesharnaisStruth2008}
J.~Desharnais and G.~Struth.
\newblock Modal semirings revisited.
\newblock In {\em International Conference on Mathematics of Program
  Construction}, pages 360--387. Springer, 2008.

\bibitem{DesharnaisStruth2011}
J.~Desharnais and G.~Struth.
\newblock Internal axioms for domain semirings.
\newblock {\em Science of Computer Programming}, 76(3):181--203, 2011.
\newblock Special issue on the Mathematics of Program Construction (MPC 2008).
\newblock \href {https://doi.org/10.1016/j.scico.2010.05.007}
  {\path{doi:10.1016/j.scico.2010.05.007}}.

\bibitem{EhmEtAl2003}
T.~Ehm, B.~M{\"o}ller, and G.~Struth.
\newblock Kleene modules.
\newblock In a.~M.~B. Berghammer, R. and G.~Struth, editors, {\em International
  Conference on Relational Methods in Computer Science}, pages 112--123.
  Springer, 2003.
\newblock \href {https://doi.org/10.1007/978-3-540-24771-5_10}
  {\path{doi:10.1007/978-3-540-24771-5_10}}.

\bibitem{FischerLadner1979}
M.~J. Fischer and R.~E. Ladner.
\newblock {Propositional dynamic logic of regular programs}.
\newblock {\em Journal of Computer and System Sciences}, 18:194--211, 1979.
\newblock \href {https://doi.org/10.1016/0022-0000(79)90046-1}
  {\path{doi:10.1016/0022-0000(79)90046-1}}.

\bibitem{GomesStruth2016}
V.~Gomes and G.~Struth.
\newblock Modal {K}leene algebra applied to program correctness.
\newblock In {\em International Symposium on Formal Methods}, pages 310--325.
  Springer, 2016.
\newblock \href {https://doi.org/10.1007/978-3-319-48989-6_19}
  {\path{doi:10.1007/978-3-319-48989-6_19}}.

\bibitem{GroenendijkStokhof1991}
J.~Groenendijk and M.~Stokhof.
\newblock Dynamic predicate logic.
\newblock {\em Linguistics and Philosophy}, 14(1):39--100, 1991.
\newblock \href {https://doi.org/10.1007/BF00628304}
  {\path{doi:10.1007/BF00628304}}.

\bibitem{HarelEtAl2000}
D.~Harel, D.~Kozen, and J.~Tiuryn.
\newblock {\em {Dynamic Logic}}.
\newblock MIT Press, 2000.

\bibitem{Hollenberg1997}
M.~Hollenberg.
\newblock An equational axiomatization of dynamic negation and relational
  composition.
\newblock {\em Journal of Logic, Language and Information}, 6(4):381--401,
  1997.
\newblock \href {https://doi.org/10.1023/A:1008271805106}
  {\path{doi:10.1023/A:1008271805106}}.

\bibitem{Hollenberg1998}
M.~Hollenberg.
\newblock Equational axioms of test algebra.
\newblock In M.~Nielsen and W.~Thomas, editors, {\em International Workshop on
  Computer Science Logic. CSL 1997}, pages 295--310, Berlin, Heidelberg, 1998.
  Springer.
\newblock \href {https://doi.org/10.1007/BFb0028021}
  {\path{doi:10.1007/BFb0028021}}.

\bibitem{Kozen1979}
D.~Kozen.
\newblock On the duality of dynamic algebras and {K}ripke models.
\newblock In E.~Engeler, editor, {\em {Proc. 1st Workshop on Logic of
  Programs}}, volume 125 of {\em Lecture Notes in Computer Science}, pages
  1--11. Springer-Verlag, 1979.

\bibitem{Kozen1994}
D.~Kozen.
\newblock A completeness theorem for {K}leene algebras and the algebra of
  regular events.
\newblock {\em Information and Computation}, 110(2):366 -- 390, 1994.
\newblock \href {https://doi.org/10.1006/inco.1994.1037}
  {\path{doi:10.1006/inco.1994.1037}}.

\bibitem{Kozen1997}
D.~Kozen.
\newblock Kleene algebra with tests.
\newblock {\em ACM Trans. Program. Lang. Syst.}, 19(3):427–443, May 1997.
\newblock \href {https://doi.org/10.1145/256167.256195}
  {\path{doi:10.1145/256167.256195}}.

\bibitem{KozenParikh1981}
D.~Kozen and R.~Parikh.
\newblock An elementary proof of the completeness of {PDL}.
\newblock {\em Theoretical Computer Science}, 14:113--118, 1981.
\newblock \href {https://doi.org/10.1016/0304-3975(81)90019-0}
  {\path{doi:10.1016/0304-3975(81)90019-0}}.

\bibitem{KozenSmith1997}
D.~Kozen and F.~Smith.
\newblock Kleene algebra with tests: Completeness and decidability.
\newblock In D.~van Dalen and M.~Bezem, editors, {\em Computer Science Logic},
  pages 244--259, Berlin, Heidelberg, 1997. Springer Berlin Heidelberg.
\newblock \href {https://doi.org/10.1145/256167.256195}
  {\path{doi:10.1145/256167.256195}}.

\bibitem{KozenTiuryn2003}
D.~Kozen and J.~Tiuryn.
\newblock Substructural logic and partial correctness.
\newblock {\em ACM Trans. Computational Logic}, 4(3):355--378, July 2003.
\newblock \href {https://doi.org/10.1145/772062.772066}
  {\path{doi:10.1145/772062.772066}}.

\bibitem{Mbacke2018}
S.~D. Mbacke.
\newblock {Completeness for Domain Semirings and Star-continuous {K}leene
  Algebras with Domain}.
\newblock {M.Sc.} thesis, Universit\'e Laval, 2018.

\bibitem{McLean2020}
B.~McLean.
\newblock Free {K}leene algebras with domain.
\newblock {\em Journal of Logical and Algebraic Methods in Programming},
  117:100606, 2020.
\newblock \href {https://doi.org/10.1016/j.jlamp.2020.100606}
  {\path{doi:10.1016/j.jlamp.2020.100606}}.

\bibitem{MoellerStruth2006}
B.~Möller and G.~Struth.
\newblock Algebras of modal operators and partial correctness.
\newblock {\em Theoretical Computer Science}, 351(2):221 -- 239, 2006.
\newblock Algebraic Methodology and Software Technology.
\newblock \href {https://doi.org/https://doi.org/10.1016/j.tcs.2005.09.069}
  {\path{doi:https://doi.org/10.1016/j.tcs.2005.09.069}}.

\bibitem{Pratt1991a}
V.~Pratt.
\newblock Dynamic algebras: {E}xamples, constructions, applications.
\newblock {\em Studia Logica}, 50(3):571--605, 1991.
\newblock \href {https://doi.org/10.1007/BF00370685}
  {\path{doi:10.1007/BF00370685}}.

\bibitem{Pratt1979}
V.~R. Pratt.
\newblock Models of program logics.
\newblock In {\em {Proceedings of the 20th Annual Symposium on Foundations of
  Computer Science}}, {FOCS ’79}, page 115–122, USA, 1979. IEEE Computer
  Society.
\newblock \href {https://doi.org/10.1109/SFCS.1979.24}
  {\path{doi:10.1109/SFCS.1979.24}}.

\bibitem{Sedlar2023}
I.~Sedl\'{a}r.
\newblock On the complexity of {K}leene algebra with domain.
\newblock In R.~Gl\"{u}ck, L.~Santocanale, and M.~Winter, editors, {\em
  {Relational and Algebraic Methods in Computer Science (RAMiCS 2023)}}, number
  13896 in Lecture Notes in Computer Science, pages 208--223, Cham, 2023.
  Springer.
\newblock \href {https://doi.org/10.1007/978-3-031-28083-2_13}
  {\path{doi:10.1007/978-3-031-28083-2_13}}.

\bibitem{Struth2016}
G.~Struth.
\newblock On the expressive power of {K}leene algebra with domain.
\newblock {\em Information Processing Letters}, 116(4):284--288, 2016.
\newblock \href {https://doi.org/10.1016/j.ipl.2015.11.007}
  {\path{doi:10.1016/j.ipl.2015.11.007}}.

\bibitem{TrnkovaReiterman1987}
V.~Trnkov\'{a} and J.~Reiterman.
\newblock Dynamic algebras with test.
\newblock {\em Journal of Computer and System Sciences}, 35(2):229 -- 242,
  1987.
\newblock \href {https://doi.org/10.1016/0022-0000(87)90014-6}
  {\path{doi:10.1016/0022-0000(87)90014-6}}.

\end{thebibliography}


\appendix
 \section{Technical appendix}
 \setcounter{lemma}{0}
 \counterwithin{lemma}{section}
  \setcounter{definition}{0}
 \counterwithin{definition}{section}
   \setcounter{observation}{0}
 \counterwithin{observation}{section}

 \subsection{Proof of Lemma \ref{lem:LC-truth}}\label{a:LC-truth}
 
 The proof follows the general strategy of \cite{KozenParikh1981}. It will be convenient to prove our result using a canonical relational model.

\begin{definition}
Let $\Gamma$ be a finite set of parameters. The \emph{canonical relational $\Gamma$-model} is $C(\Gamma) = \langle \mathbb{C}(\Gamma), \mathsf{rel}_{C(\Gamma)}, \mathsf{sat}_{C(\Gamma)}\rangle$ where
\[ \mathsf{rel}_{C(\Gamma)}(\mathtt{a}) = \{ \langle G, H\rangle \mid G \mathtt{a} H \nequiv \mathtt{0} \}
\qquad \mathsf{sat}_{C(\Gamma)}(\mathtt{p}) = \{ G \mid G \leqq \mathtt{p} \}\] for all $\mathtt{a} \in \Sigma$ and $\mathtt{p} \in \Pi$. The $C(\Gamma)$-interpretation of expressions is defined as usual in relational models.
\end{definition}

We write $G \xto{e} H$ iff $\langle G, H\rangle \in \llbracket e\rrbracket_{C(\Gamma)}$. Note that
\begin{itemize}
\item $G \xto{\mathtt{a}} H$ iff $G \mathtt{a} H \nequiv \mathtt{0}$
\item $G \xto{\mathtt{p}} H$ iff $G = H$ and $G \leqq \mathtt{p}$
\item $G \xto{\mathtt{1}} H$ iff $G = H$
\item $G \xto{\mathtt{0}} H$ iff $H \neq H$ (i.e.~never)
\item $G \xto{e + f} H$ iff $G \xto{e} H$ or $G \xto{f} H$
\item $G \xto{e \cdot f} H$ iff there is $I$: $G \xto{e} I \xto{f} H$
\item $G \xto{e^{*}} H$ iff there are $I_1, \ldots, I_n$ such that \[ G = I_1 \xto{e} \ldots \xto{e} I_n = H\]
\item $G \xto{e^{\bot}} H$ iff $G = H$ and there is no $I$: $G \xto{e} I$
\item $G \xto{e^{\dom}} H$ iff $G = H$ and there is $I$: $G \xto{e} I$
\end{itemize}

The relation between relational and guarded-string interpretations of expressions can be compactly described using the function $\mathsf{pair} : 2^{\mathbb{CS}(\Gamma)} \to 2^{\mathbb{C}(\Gamma) \times \mathbb{C}(\Gamma)}$ defined in the obvious way:
\[ \mathsf{pair}(L) = \{ \langle G, H\rangle \mid \exists G \mathtt{a}_1 \ldots \mathtt{a}_{n-1} H \in L\}\]

\begin{lemma}\label{lem:pair}
For all $e \in \mathbb{E}$, $\llbracket e\rrbracket_{C(\Gamma)} = \mathsf{pair} \left( \llbracket e\rrbracket_{\Gamma} \right)$.
\end{lemma}

\begin{lemma}\label{lem:LC-truth-2}
Let $\Gamma$ be a FL-closed set of parameters. Then the following holds in $C(\Gamma)$ for all $e \in \mathbb{E}$:
\begin{itemize}
\item[(i)] If $e^{\top} \in \Gamma$, then $G e H \nequiv \mathtt{0}$ implies $G \xto{e} H$.
\item[(ii)] If $(e \cdot \f)^{\top} \in \Gamma$, then $G \xto{e} H \leqq \f$ implies $G \leqq (e \cdot \f)^{\top}$.
\end{itemize}
\end{lemma}
\begin{proof}
Both claims are established simultaneously by induction on the structure of $e$. The base case (i) for $\mathtt{a} \in \Sigma$ holds by definition. The base case (i) for $\mathtt{p} \in \Pi$ holds since $G \mathtt{p} \nequiv \mathtt{0}$ implies (a) $G \tin \mathtt{p}$ ($\mathtt{p}^{\dom} \in \Gamma$ only if $\mathtt{p} \in \Gamma$) and so (b) $G \cdot H \nequiv \mathtt{0}$ which is possible only if $G = H$. Base case (ii) for $\mathtt{a} \in \Sigma$: If $G \xto{\mathtt{a}} H \leqq \f$, then $G \mathtt{a} H \nequiv \mathtt{0}$ and $H \leqq \f$. Then $G (\mathtt{a} \cdot \f) \not\leqq \mathtt{0}$, which means that $G \not\leqq (\mathtt{a} \cdot \f)^{\bot} \equiv (\mathtt{a} \cdot \f)^{\top\ant}$. Hence, $G \leqq (\mathtt{a} \cdot \f)^{\top}$. Base case (ii) for $\mathtt{p} \in \Pi$: $G \xto{\mathtt{p}} H$ means that $G = H$ and $H \leqq \mathtt{p}$. If also $H \leqq \f$, then $G \leqq (\mathtt{p} \cdot \f) \equiv (\mathtt{p} \cdot \f)^{\top}$.

Claim (i) for $\mathtt{1}$: If $G \mathtt{1} H \nequiv \mathtt{0}$, then $GH \nequiv \mathtt{0}$, but this is possible only if $G = H$. Claim (ii) for $\mathtt{1}$: If $G \xto{\mathtt{1}} H$, then $G = H$. If also $H \leqq \f$, then $G \leqq (\mathtt{1} \cdot \f) \equiv (\mathtt{1} \cdot \f)^{\top}$. Claims (i) and (ii) for $\mathtt{0}$ hold vacuously.

The induction hypothesis is that both (i) and (ii) hold for $e, f$. We will prove (i) and (ii) for $e + f$, $e \cdot f$, $e^{*}$, $e^{\dom}$ and $e^{\bot}$.

Induction step (i) for $e + f$: If $(e + f)^{\top} \in \Gamma$, then $e^{\top} \in \Gamma$ and $f^{\top} \in \Gamma$. If $G(e + f)H \nequiv \mathtt{0}$, then $GeH \nequiv \mathtt{0}$ or $GfH \nequiv \mathtt{0}$ by distributivity. Hence, $G \xto{e} H$ or $G \xto{f} H$ by the induction hypothesis. It follows that $G \xto{e + f} H$.

Induction step (i) for $e \cdot f$: If $(e \cdot f)^{\top} \in \Gamma$, then $e^{\top} \in \Gamma$ and $f^{\top} \in \Gamma$. If $G(e \cdot f) H \nequiv \mathtt{0}$, then there is $I$ such that $G e I \nequiv \mathtt{0}$ and $I f H \nequiv \mathtt{0}$.\footnote{If for all $I$ either $GeI \equiv \mathtt{0}$ or $I f H \equiv \mathtt{0}$, then $GeIfH \equiv \mathtt{0}$ for all $I \in \mathbb{C}(\Gamma)$. Then $Ge \mathbb{C}(\mathtt{1}^{\top}) f H \equiv \mathtt{0}$ and so $G e \mathtt{1} f H \equiv \mathtt{0}$ by Lemma \ref{lem:LC-truth-1} and $\mathtt{1}^{\top} \equiv \mathtt{1}$. This means that $G(e \cdot f)H \equiv \mathtt{0}$ and we are done.} Using the induction hypothesis we infer that $G \xto{e} I \xto{f} H$, and so $G \xto{e \cdot f} H$.

Induction step (i) for $e^{*}$: If $(e^{*})^{\top} \in \Gamma$, then $e^{\top} \in \Gamma$. Assume that $G e^{*} H \nequiv \mathtt{0}$. Let $X$ be the smallest subset of $\mathbb{C}(\Gamma)$ such that
\begin{itemize}
\item $G \in X$
\item $I \in X$ and $I e J \nequiv \mathtt{0}$ only if $J \in X$ 
\end{itemize}
Since $e^{\top} \in \Gamma$, we may use the induction hypothesis to show that $I \xto{e} J$ for all $I, J$ satisfying the above conditions. From this we infer that $G \xto{e^{*}} J$ for all $J \in X$. Hence, we have to show that $H \in X$. We will write $X$ instead of $\sum X$. It is sufficient to show that $H \cdot X \nequiv \mathtt{0}$.\footnote{If $H \cdot \sum_{J \in X} J \nequiv \mathtt{0}$, then there is $J \in X$ such that $H \cdot J \nequiv \mathtt{0}$. If $H \cdot J \nequiv \mathtt{0}$, then $H = J$ by \eqref{eq:ECQ-bar}.} If $X e X^{\bot} \nequiv \mathtt{0}$, then there are $I \in X$ and $J \notin X$ such that $I e J \nequiv \mathtt{0}$.\footnote{If $X e X^{\bot} \nequiv \mathtt{0}$, then there is $I \in X$ such that $I e X^{\bot} \nequiv \mathtt{0}$. If $I e J \equiv \mathtt{0}$ for all $J \notin X$, then $I e \left( \sum_{J \notin X} J \right) \equiv \mathtt{0}$ and so $I e X^{\bot} \equiv \mathtt{0}$, contrary to what we have observed, since $X^{\bot} \leqq \sum_{J \notin X} J$. This last claim is established by noting that if $X^{\bot} \not\leqq \sum_{J \notin X} J$, then $G^{\bot}_1 \ldots G_n^{\bot} \nequiv \mathtt{0}$ for $\mathbb{C}(\Gamma) = \{ G_1, \ldots, G_n \}$. However, then $\mathbb{C}(\mathtt{1})^{\bot} \nequiv \mathtt{0}$, contradicting Lemma \ref{lem:LC-truth-1}.} This contradicts our assumptions about $X$, and so $X e X^{\bot} \equiv \mathtt{0}$. Hence, $X \leqq (e \cdot X^{\bot})^{\bot}$ and so $X \leqq (e^{*} \cdot X^{\bot})^{\bot}$ by \eqref{eq:Star-2}. Since $G \leqq X$, we obtain $G e^{*} X^{\bot} \equiv \mathtt{0}$. Using the assumption $G e^{*} H \nequiv \mathtt{0}$, we conclude that $H \not\leqq X^{\bot}$ and so $H \cdot X \nequiv \mathtt{0}$.

The induction step (i) for $e^{\dom}$ holds vacuously ($e^{\dom\dom}$ is not a parameter). Induction step (i) for $e^{\bot}$: If $e^{\bot\top} \in \Gamma$, then $(e \cdot \mathtt{1})^{\top} \in \Gamma$. Assume $Ge^{\bot}H \nequiv \mathtt{0}$. Then $E \cdot H \nequiv \mathtt{0}$, which is possible only in case $G = H$. Now suppose that there is $I$ such that $G \xto{e} I$. Since $I \leqq \mathtt{1}$ and $(e \cdot \mathtt{1})^{\top} \in \Gamma$, the induction hypothesis (ii) entails $G \leqq (e \cdot \mathtt{1})^{\top} \equiv e^{\top}$. But then $G e^{\bot} \leqq e^{\top} e^{\bot} \leqq \mathtt{0}$, which contradicts the assumption that $Ge^{\bot}H \nequiv \mathtt{0}$. Hence, there is no $I$ such that $G \xto{e} I$, which means that $G \xto{e^{\bot}} G = H$.

Now we consider the induction step (ii). Induction step (ii) for $e + f$: If $G \xto{e+f} H$, then $G \xto{e} H$ or $G \xto{f} H$. Since $\Gamma$ is FL-closed, $((e + f) \cdot \f)^{\top}$ entails $(e \cdot \f)^{\top} \in \Gamma$ and $(f \cdot \f)^{\top} \in \Gamma$. If $H \leqq \f$, then it suffices to use the induction hypothesis in both cases $G \xto{e} H$ or $G \xto{f} H$ to obtain $G \leqq (e \cdot \f)^{\top} + (f \cdot \f)^{\top} \equiv ((e + f) \cdot \f)^{\top}$.

Induction step (ii) for $e \cdot f$: If $G \xto{e \cdot f} H$, then there is $I$ such that $G \xto{e} I$ and $I \xto{f} H$. If $H \leqq \f$, then we reason as follows. Since $\Gamma$ is FL-closed, $((e \cdot f) \cdot \f)^{\top} \in \Gamma$ entails $(f \cdot \f)^{\top} \in \Gamma$ and $(e \cdot (f \cdot \f)^{\ant\ant})^{\dom} \in \Gamma$. Using the induction hypothesis, we obtain $I \leqq (f \cdot \f)^{\top} \Eq (f \cdot \f)^{\ant\ant} $ and also $G \leqq (e \cdot (f \cdot \f)^{\ant\ant})^{\top} \equiv ((e \cdot f) \cdot \f)^{\top}$.

Induction step (ii) for $e^{*}$: If $G \xto{e^{*}} H$, then there are $I_1, \ldots, I_n$ such that $G = I_1 \xto{e} \ldots \xto{e} I_n = H$. We show that if $H \leqq \f$, then $I_m \leqq (e^{*} \cdot \f)^{\top}$ for all $m \leq n$. First, $I_n \leqq \f \equiv G \leqq \mathtt{1} \cdot \f \equiv (\mathtt{1} \cdot \f)^{\top} \leqq (e^{*} \cdot \f)^{\top}$. Second, assume that $I_{k+1} \leqq (e^{*} \cdot \f)^{\top} \Eq (e^{*} \cdot \f)^{\ant\ant}$. Since $\Gamma$ is FL-closed, we have $(e \cdot (e^{*} \cdot \f)^{\ant\ant})^{\top} \in \Gamma$. Hence, using the induction hypothesis we obtain $I_{k} \leqq (e \cdot (e^{*} \cdot \f)^{\ant\ant})^{\top} \equiv (e \cdot e^{*} \cdot \f)^{\top} \leqq (e^{*} \cdot \f)^{\top}$.

The induction step (ii) for $e^{\top}$ holds vacuously since $(e^{\top} \cdot \f)^{\top}$ is not a parameter. Induction step (ii) for $e^{\bot}$: If $G \xto{e^{\bot}} H$, then $G = H$ and there is no $I$ such that $G \xto{e} I$. Using Lemma \ref{lem:LC-truth-1} and the induction hypothesis (i) we obtain $G e \equiv \mathtt{0}$, and so $G \leqq e^{\bot}$. If $H \leqq \f$, then $G \leqq (e^{\bot} \cdot \f) \equiv (e^{\bot} \cdot \f)^{\top}$.
\end{proof}

\noindent
\textbf{Lemma \ref{lem:LC-truth}.} 
\textit{If $\Gamma$ is a FL-closed set of parameters, then $\f \in \Gamma$ only if
\[ \llbracket \f\rrbracket_{\Gamma} = \{ G \in \mathbb{C}(\Gamma) \mid G \leqq \f \}\]}
\noindent\vspace*{-8mm}
\begin{proof}
The claim for $\f \in \Pi$ holds by definition. Hence, it remains to show that that if $e^{\top} \in \Gamma$, then $\llbracket e^{\top}\rrbracket_{\Gamma} = \{ G \in \mathbb{C}(\Gamma) \mid G \leqq e^{\top}\}$. In particular, using the definition of a canonical $\Gamma$-interpretation, we have to show that, for all $G \in \mathbb{C}(\Gamma)$
\[ G \leqq e^{\top} \iff \exists H. G \xto{e} H \]
Left to right: If $G \leqq e^{\top}$, then $G e H \nequiv \mathtt{0}$ for some $H$ since otherwise we would have $G e \left( \sum \mathbb{C}(\Gamma) \right) \equiv \mathtt{0}$ which entails $G \leqq e^{\bot}$ by Lemma \ref{lem:LC-truth-1}. This would mean that $G$ is inconsistent. Hence, $GeH \nequiv \mathtt{0}$ for some $H$ and so we may infer $G \xto{e} H$ using Lemma \ref{lem:LC-truth-2}(i). Right to left: Assume that $G \xto{e} H$. If $e \neq (f \cdot \mathtt{1})$ for all $f$, then $e^{\top} \in \Gamma$ implies $(e \cdot \mathtt{1})^{\top} \in \Gamma$. Hence, $G \leqq (e \cdot \mathtt{1})^{\top} \equiv e^{\top}$ by Lemma \ref{lem:LC-truth-2}(ii) since obviously $H \leqq \mathtt{1}$. If $e = (f \cdot \mathtt{1})$ for some $f$, then $G \xto{e} H$ entails $G \xto{f} H$. Since $(f \cdot \mathtt{1})^{\top}$ in that case, we may use Lemma \ref{lem:LC-truth-2}(ii) to infer $G \leqq (f \cdot \mathtt{1})^{\top} = e^{\top}$.
\end{proof}

\subsection{Proof of Lemma \ref{lem:RC-2}}\label{a:RC}

The relational model $CS(\Gamma)$ is what modal logicians call an ``unravelling'' of the model $C(\Gamma)$ discussed in Section \ref{a:LC-truth}. As such, $CS(\Gamma)$ and $C(\Gamma)$ are related in a useful way.

\begin{definition}\label{def:bisim}
Let $M_1 = \langle X_1, \mathsf{rel}_1, \mathsf{sat}_1\rangle$ and $M_2 = \langle X_2, \mathsf{rel}_2, \mathsf{sat}_2\rangle$ be two relational models. A binary relation $\sim$ between $X_1$ and $X_2$ is a \emph{bisimulation} between $M_1$ and $M_2$ iff $x_1 \sim x_2$ implies
\begin{enumerate}
\item  $x_1 \in \mathsf{sat}_1(\mathtt{p}) \iff x_2 \in \mathsf{sat}_2(\mathtt{p})$ for all $\mathtt{p}$
\item $\langle x_1, y_1\rangle \in \mathsf{rel}_1 (\mathtt{a})$ only if there is $y_2 \in X_2$ such that $\langle x_2, y_2\rangle \in \mathsf{rel}_2(\mathtt{a})$ and $y_1 \sim y_2$, for all $\mathtt{a}$
\item $\langle x_2, y_2\rangle \in \mathsf{rel}_2 (\mathtt{a})$ only if there is $y_1 \in X_1$ such that $\langle x_1, y_1\rangle \in \mathsf{rel}_1(\mathtt{a})$ and $y_1 \sim y_2$, for all $\mathtt{a}$
\end{enumerate}
If there is a bisimulation between $M_1$ and $M_2$, then we say that the models are \emph{bisimilar}. If there is a bisimulation $\sim$ between $M_1$ and $M_2$ such that $x_1 \sim x_2$, then we write $(M_1, x_1) \bisim (M_2, x_2)$.
\end{definition}

The following is a standard result in modal logic, transposed to our setting. Recall that $(M, x) \vDash \f$ is our notation for $\langle x, x\rangle \in \llbracket \f\rrbracket_M$.

\begin{lemma}\label{lem:bisim-1}
If $(M, x) \bisim (N, y)$, then $(M, x) \vDash \f \iff (N, y) \vDash \f$ for all $\f \in \mathbb{F}$.
\end{lemma}

\begin{lemma}\label{lem:bisim-2}
For each set of tests $\Gamma$, the relation $\sim$ between $\mathbb{C}(\Gamma)$ and $\mathbb{CS}(\Gamma)$ defined by
\[ G \sim H_1 \mathtt{a}_1 \ldots \mathtt{a}_{n-1} H_n \iff G = H_n\] is a bisimulation between $C(\Gamma)$ and $CS(\Gamma)$.
\end{lemma}

\noindent
\textbf{Lemma \ref{lem:RC-2}.}
\textit{For all $\f \in \Gamma$, $\llbracket \f\rrbracket_M = \{ \langle w, w\rangle \mid \mathsf{last}(w) \leqq \f \}$.}
\begin{proof}
\begin{align*}
(CS(\Gamma), w) \vDash \f 
	& \iff (C(\Gamma), \mathsf{last}(w)) \vDash \f 
		&& \text{by Lemmas \ref{lem:bisim-1} and \ref{lem:bisim-2}}\\
	& \iff \mathsf{last}(w) \leqq \f
		&& \text{by Lemmas \ref{lem:LC-truth} and \ref{lem:pair}}
\end{align*}
\end{proof}

\end{document}